\documentclass[a4paper,10pt]{article}
\usepackage{test,caption,subcaption}

\usepackage[authoryear,round]{natbib}
\usepackage[colorlinks,citecolor={blue}]{hyperref}
\usepackage{doi}
\usepackage{xspace}

\usepackage[inline,shortlabels]{enumitem}
\setlist[enumerate,1]{label=(\roman*)}

\title{A Theory of the Saving Rate of the Rich\thanks{We thank Chris Carroll, {\'E}milien Gouin-Bonenfant, Ben Moll, Johannes Wieland, and seminar participants at UCSD and 2020 CRETA Economic Theory Conference for valuable feedback and suggestions. A previous version of this paper was circulated under the title ``Asymptotic Marginal Propensity to Consume''.}}

\author{Qingyin Ma\thanks{International School of Economics and Management, Capital University of Economics and Business. Email: \href{mailto:qingyin.ma@cueb.edu.cn}{qingyin.ma@cueb.edu.cn}.} \and Alexis Akira Toda\thanks{Department of Economics, University of California San Diego. Email: \href{mailto:atoda@ucsd.edu}{atoda@ucsd.edu}.}}

\numberwithin{equation}{section}
\numberwithin{thm}{section}
\numberwithin{exmp}{section}

\newcommand{\cC}{\mathcal{C}}
\newcommand{\cH}{\mathcal{H}}
\newcommand{\ZZ}{\mathsf{Z}}

\begin{document}

\maketitle

\begin{abstract}
Empirical evidence suggests that the rich have higher propensity to save than do the poor. While this observation may appear to contradict the homotheticity of preferences, we theoretically show that that is not the case. Specifically, we consider an income fluctuation problem with homothetic preferences and general shocks and prove that consumption functions are asymptotically linear, with an exact analytical characterization of asymptotic marginal propensities to consume (MPC). We provide necessary and sufficient conditions for the asymptotic MPCs to be zero. %We solve a calibrated model with standard constant relative risk aversion utility and show that asymptotic MPCs can be zero in empirically plausible settings, implying an increasing and large saving rate of the rich and high wealth inequality.
We calibrate a model with standard constant relative risk aversion utility and show that zero asymptotic MPCs are empirically plausible, implying that our mechanism has the potential to accommodate a large saving rate of the rich and high wealth inequality (small Pareto exponent) as observed in the data.

\medskip


\medskip

\end{abstract}

\section{Introduction}

Empirical evidence suggests that the rich have higher propensity to save than do the poor.\footnote{\cite*{Quadrini1999} documents that entrepreneurs (who tend to be rich) have high saving rates. \cite*{DynanSkinnerZeldes2004} document that there is a positive association between saving rates and lifetime income. More recently, using Norwegian administrative data, \cite*{FagerengHolmMollNatvikWP} show that among households with positive net worth, saving rates are increasing in wealth.} This fact implies that the rich have lower marginal propensity to consume (MPC), which has important economic consequences. For example, when the rich have lower MPC, the consumption tax, which is a popular tax instrument in many countries, becomes regressive and may not be desirable from equity perspectives. MPC heterogeneity also implies that the wealth distribution matters for determining aggregate demand and hence monetary and fiscal policies \citep*{KaplanMollViolante2018,MianStraubSufiWP}. 

Why do the rich save so much? Intuition suggests that canonical models of consumption and savings that feature (identical) homothetic preferences are unable to explain the high saving rate of the rich: in such models, consumption (hence saving) functions should be asymptotically linear in wealth due to homotheticity, implying an asymptotically constant saving rate. A seemingly obvious explanation for the high saving rate of the rich is that preferences are not homothetic.\footnote{For example, \cite*{Carroll2000Why} considers a `capitalist spirit' model in which agents directly get utility from holding wealth, where the utility functions for consumption and wealth have different curvatures. \cite*{DeNardi2004} considers a model with bequest, which is mathematically similar. \cite*{StraubSavingWP} estimates that the elasticity of consumption with respect to permanent income is below 1 (which implies concavity of consumption functions) and uses non-homothetic preferences to explain it. Another possibility is to introduce frictions such as portfolio adjustment costs \citep*{FagerengHolmMollNatvikWP}.} However, non-homothetic preferences have some undesirable theoretical properties. First, they are inconsistent with balanced growth (whereas many aggregate economic variables such as real per capita GDP are near unit root processes), at least in basic models in which preference parameters are constant. Second, non-homothetic utility functions have more parameters than homothetic ones, which introduces arbitrariness in model specification and calibration.

In this paper we theoretically show that the intuition of ``homotheticity implies (asymptotic) linearity'' is only partially correct. We consider a standard income fluctuation problem with (homothetic) constant relative risk aversion (CRRA) preference but with capital and labor income risks in a general Markovian setting. We prove that the consumption functions are asymptotically linear in wealth, or the asymptotic marginal propensities to consume converge to some constants.\footnote{Throughout the paper we say that a consumption function $c(a)$ (where $a>0$ is financial wealth) is \emph{asymptotically linear} if the asymptotic \emph{average} propensity to consume $\bar{c}=\lim_{a\to\infty}c(a)/a$ exists. This condition is weaker than $\lim_{a\to\infty}\abs{c(a)-\bar{c}a-d}=0$ for some $\bar{c},d\in \R$, which may be a more common definition of asymptotic linearity. If the asymptotic MPC $\bar{c}=\lim_{a\to\infty}c'(a)$ exists, then l'H\^opital's rule implies $\lim_{a\to\infty}c(a)/a=\lim_{a\to\infty}c'(a)=\bar{c}$. Although not necessarily mathematically precise, due to the lack of better language we use ``constant asymptotic average propensity to consume'', ``constant asymptotic MPC'', and ``asymptotic linearity'' interchangeably.} While this statement is intuitive, there is one surprise: we obtain an exact analytical characterization of the asymptotic MPCs and prove that they can be zero. The asymptotic MPCs depend only on risk aversion and the stochastic processes for the discount factor and return on wealth, and are independent of the income process. Furthermore, we derive necessary and sufficient conditions for zero asymptotic MPCs. When the asymptotic MPCs are zero, the saving rates of the rich converge to one as agents get wealthier. Thus, we provide a potential explanation for why the rich save so much, and we do so with standard homothetic preferences.

To prove that consumption functions are asymptotically linear with particular slopes, we apply policy function iteration as in \cite*{LiStachurski2014} and \cite*{MaStachurskiToda2020JET}. Since agents cannot consume more than their financial wealth in the presence of borrowing constraints, a natural upper bound on consumption is asset, which is linear with a slope of 1. Starting from this candidate consumption function, policy function iteration results in increasingly tighter upper bounds. On the other hand, we directly obtain lower bounds by restricting the space of candidate consumption functions such that they have linear lower bounds with specific slopes. We analytically derive these slopes based on the fixed point theory of monotone convex maps developed in \cite*{Du1990}, which has recently been applied in economics by \cite*{Toda2019JME} and \cite*{BorovickaStachurski2020}. Finally, we show that the upper and lower bounds thus obtained have identical slopes, implying the asymptotic linearity of consumption functions with an exact characterization of asymptotic MPCs.

To assess the empirical plausibility of our new mechanism, we numerically solve a partial equilibrium model with CRRA utility and capital income risk calibrated to the U.S.\ economy. We find that with moderate risk aversion (above 3), the asymptotic MPCs become zero, the saving rates of the rich are increasing and approach 1, and the implied wealth Pareto exponent is close to the value in the data.

%\subsection{Related literature}

Our paper is related to the theoretical studies of the income fluctuation problem, which is a key building block of heterogeneous-agent models in modern macroeconomics.\footnote{See, for example, \cite*{Cao2020} and \cite*{Acikgoz2018} for the existence of equilibrium with and without aggregate shocks, where the theoretical properties of the income fluctuation problem play an important role. \cite*{LehrerLight2018} and \cite*{Light2018} prove comparative statics results regarding savings. \cite*{Light2020} proves the uniqueness of stationary equilibrium in an Aiyagari model that exhibits a certain gross substitute property.} \cite*{ChamberlainWilson2000} study the existence of a solution assuming bounded utility and applying the contraction mapping theorem. \cite*{LiStachurski2014} relax the boundedness assumption and apply policy function iteration. \cite*{BenhabibBisinZhu2015} consider a special model with CRRA utility, constant discounting, and \iid and mutually independent returns and income shocks to study the tail behavior of wealth. \cite*{MaStachurskiToda2020JET} allow for stochastic discounting and returns on wealth in a general Markovian setting and discuss the ergodicity, stochastic stability, and tail behavior of wealth. \cite*{Carroll2020} examines detailed properties of a special model with CRRA utility, constant discounting and risk-free rate, and \iid permanent and transitory income shocks.

While the main focus of these papers is the existence, uniqueness, and computation of a solution, we focus on the asymptotic behavior of consumption with general shocks. \cite*{CarrollKimball1996} show the concavity of consumption functions in a class of income fluctuation problems with hyperbolic absolute risk aversion (HARA) utility,\footnote{\cite*{TodaHARA} shows that HARA is necessary for the concavity of consumption functions.} which implies asymptotic linearity. However, they do not characterize the asymptotic MPCs as we do. As an intermediate result to examine the wealth accumulation process, Proposition~5 of \cite*{BenhabibBisinZhu2015} characterizes the asymptotic MPC of a special model described above. \cite*{Carroll2020} also intuitively discusses the asymptotic linearity of the consumption function in a model without capital income risk, and points out in Appendix~A.2.2 and Figure~6 the possibility of zero asymptotic MPCs, although that case requires a negative interest rate. Our contribution relative to these results is that we obtain a rigorous and complete characterization of asymptotic MPCs in a general setting (including capital income risk and Markovian shocks), %provide mathematically rigorous proofs, 
analyze the necessity of these advanced features in generating zero asymptotic MPCs, and show through a numerical example that they are empirically plausible.

The rest of the paper is organized as follows. %After a brief discussion of the related literature, 
Section~\ref{sec:AL} introduces a general income fluctuation problem, proves the asymptotic linearity of consumption functions with homothetic preferences, and discusses some examples. Section~\ref{sec:plausible} applies the theory to a calibrated model and shows that zero asymptotic MPCs are empirically plausible and generate high wealth inequality (small Pareto exponent). Appendices~\ref{sec:IF} and \ref{sec:proof} contain the proofs. Replication files for Section~\ref{sec:plausible} are available at \url{https://github.com/alexisakira/savingrate}.

\section{Asymptotic linearity of consumption functions}\label{sec:AL}

In this section we introduce an income fluctuation problem that generalizes the setting in \cite*{MaStachurskiToda2020JET} and study the asymptotic property of the consumption functions when preferences are homothetic.

\subsection{Income fluctuation problem}
Time is discrete and denoted by $t=0,1,2,\dotsc$. Let $a_t$ be the financial wealth of the agent at the beginning of period $t$. The agent chooses consumption $c_t\ge 0$ and saves the remaining wealth $a_t-c_t$. The period utility function is $u$ and the discount factor, gross return on wealth, and non-financial income in period $t$ are denoted by $\beta_t,R_t,Y_t$, where we normalize $\beta_0=1$. Thus the agent solves
\begin{subequations}\label{eq:IF}
\begin{align}
&\maximize && \E_0\sum_{t=0}^\infty \left(\prod_{i=0}^t\beta_i\right)u(c_t) \notag\\
&\st && a_{t+1}=R_{t+1}(a_t-c_t)+Y_{t+1}, \label{eq:IF_budget}\\
&&& 0\le c_t\le a_t, \label{eq:IF_borrow}
\end{align}
\end{subequations}
where the initial wealth $a_0=a>0$ is given, \eqref{eq:IF_budget} is the budget constraint, and \eqref{eq:IF_borrow} implies that the agent cannot borrow.\footnote{The no-borrowing condition $a_t-c_t\ge 0$ is without loss of generality as discussed in \cite*{ChamberlainWilson2000} and \cite*{LiStachurski2014}.} The stochastic processes $\set{\beta_t,R_t,Y_t}_{t\ge 1}$ obey
\begin{equation}
\beta_t=\beta(Z_{t-1},Z_t,\zeta_t),\quad R_t=R(Z_{t-1},Z_t,\zeta_t),\quad Y_t=Y(Z_{t-1},Z_t,\zeta_t),\label{eq:betaRY}
\end{equation}
where $\beta,R,Y$ are nonnegative measurable functions, $\set{Z_t}_{t\ge 0}$ is a time-homogeneous Markov chain taking values in a finite set $\ZZ=\set{1,\dots,Z}$ with a transition probability matrix $P$, and the innovations $\set{\zeta_t}$ are independent and identically distributed (\iid) over time and could be vector-valued.

Before discussing the properties of the income fluctuation problem \eqref{eq:IF}, we note that it is very general despite the fact that it is cast as an infinite-horizon optimization problem in a stationary environment. For example, a finite lifetime is permitted by allowing $\beta(Z_{t-1},Z_t,\zeta_t)=0$ in some states. Life-cycle features such as age-dependent income and mortality risk \citep*{huggett1996} are also permitted by supposing that agents have some finite upper bound for age, time is part of the state $z\in \ZZ$, and that the discount factor $\beta(Z_{t-1},Z_t,\zeta_t)$ includes survival probability.

To simplify the notation, we introduce the following conventions. We use a hat to denote a random variable that is realized next period, for example $Z=Z_t$ and $\hat{Z}=Z_{t+1}$. When no confusion arises, we write $\hat{\beta}$ for $\beta(Z,\hat{Z},\hat{\zeta})$ and define $\hat{R},\hat{Y}$ analogously. Conditional expectations are abbreviated using subscripts, for example
\begin{equation*}
\E_zX=\CE{X|Z=z}\quad \text{and}\quad \E_{z,\hat{z}}X=\CE{X|Z=z,\hat{Z}=\hat{z}}.
\end{equation*}
For $\theta\in \R$, we define the matrix $K(\theta)$ related to the transition probability matrix $P$, discount factor $\beta$, and return $R$ by
\begin{equation}
K_{z\hat{z}}(\theta)\coloneqq P_{z\hat{z}}\E_{z,\hat{z}}\hat{\beta}\hat{R}^\theta=P_{z\hat{z}}\E\beta(z,\hat{z},\hat{\zeta})R(z,\hat{z},\hat{\zeta})^\theta\in [0,\infty]. \label{eq:Ktheta}
\end{equation}
The matrix $K(\theta)$ for various values of $\theta$ appears throughout the paper. For a square matrix $A$, the scalar $r(A)$ denotes its spectral radius (largest absolute value of all eigenvalues), \ie,
\begin{equation}
r(A)\coloneqq \max \set{\abs{\alpha}|\text{$\alpha$ is an eigenvalue of $A$}}.\label{eq:spectralRadius}
\end{equation}
The spectral radius \eqref{eq:spectralRadius} plays an important role in the subsequent discussion. %Finally, for two matrices $A=(A_{z\hat{z}})$ and $B=(B_{z\hat{z}})$ of the same size, $A\odot B=(A_{z\hat{z}}B_{z\hat{z}})$ denotes the Hadamard (entry-wise) product.

Consider the following assumptions.

\begin{asmp}\label{asmp:Inada}
The utility function $u:[0,\infty)\to \R \cup \set{-\infty}$ is continuously differentiable on $(0,\infty)$, $u'$ is positive and strictly decreasing on $(0,\infty)$, and $u'(\infty)<1$.
\end{asmp}

Assumption~\ref{asmp:Inada} is essentially the usual monotonicity and concavity assumptions together with a form of Inada condition ($u'(\infty)<1$).

\begin{asmp}\label{asmp:spectral}
Let $K$ be as in \eqref{eq:Ktheta}. The following conditions hold:
\begin{enumerate}
\item\label{item:Ebeta} The matrices $K(0)$ and $K(1)$ are finite,
\item\label{item:rbeta} $r(K(0))<1$ and $r(K(1))<1$,
\item\label{item:EY} $\E_{z,\hat{z}}\hat{Y}<\infty$, $\E_{z,\hat{z}}u'(\hat{Y})<\infty$, and $\E_{z,\hat{z}}\hat{\beta}\hat{R}u'(\hat{Y})<\infty$ for all $(z,\hat{z})\in \ZZ^2$.
\end{enumerate}
\end{asmp}

Using the definition of $K$ in \eqref{eq:Ktheta}, condition~\ref{item:Ebeta} is equivalent to $\E_{z,\hat{z}}\hat{\beta}<\infty$ and $\E_{z,\hat{z}}\hat{\beta}\hat{R}<\infty$ for all $(z,\hat{z})\in \ZZ^2$. The condition $r(K(0))<1$ in \ref{item:rbeta} generalizes $\beta<1$ to the case with random discount factors. The condition $r(K(1))<1$ generalizes the `impatience' condition $\beta R<1$ to the stochastic case.

Our setting and Assumptions~\ref{asmp:Inada}, \ref{asmp:spectral} are similar to those in \cite*{MaStachurskiToda2020JET} but slightly more general. They suppose that $\beta_t,R_t,Y_t$ depend only on the current state $Z_t$ and \iid innovations that are mutually independent, whereas we allow $\beta_t,R_t,Y_t$ to also depend on the previous state $Z_{t-1}$ as in \eqref{eq:betaRY} and the innovations could be correlated (because $\zeta_t$ in \eqref{eq:betaRY} is vector-valued with arbitrary distribution). Although the potential dependence on $Z_{t-1}$ is mathematically redundant because we can always square the state space as $\ZZ^2$ and define a new variable $\tilde{Z}_t\coloneqq (Z_{t-1},Z_t)$, it is computationally advantageous to reduce the dimensionality. \cite*{MaStachurskiToda2020JET} suppose that the utility function $u$ is twice continuously differentiable and $u'(0)=\infty$. Our only substantive generalization is that we allow the possibility $u'(0)<\infty$, which for example can accommodate hyperbolic absolute risk aversion utility.\footnote{In addition to the counterparts of Assumptions~\ref{asmp:Inada} and \ref{asmp:spectral}, \cite*{MaStachurskiToda2020JET} assume that the transition probability matrix $P$ is irreducible. However, irreducibility is required only for their ergodicity result, not for existence and uniqueness of a solution.}

Under the maintained assumptions, Theorem~\ref{thm:exist} below states that the income fluctuation problem \eqref{eq:IF} admits a unique solution and provides a computational algorithm. To make its statement precise, we introduce further definitions. Let $\cC$ be the space of candidate consumption functions such that $c:(0,\infty)\times \ZZ\to \R$ is continuous, is increasing in the first argument, $0\le c(a,z)\le a$ for all $a>0$ and $z\in \ZZ$, and
\begin{equation}
\sup_{(a,z)\in (0,\infty)\times \ZZ}\abs{u'(c(a,z))-u'(a)}<\infty.\label{eq:uprimediff}
\end{equation}
For $c,d\in \cC$, define the metric
\begin{equation}
\rho(c,d)=\sup_{(a,z)\in (0,\infty)\times \ZZ}\abs{u'(c(a,z))-u'(d(a,z))}.\label{eq:rho}
\end{equation}
When $u'$ is positive, continuous, and strictly decreasing (implied by Assumption~\ref{asmp:Inada}), it is straightforward (\eg, Proposition 4.1 of \cite*{LiStachurski2014}) to show that $(\cC,\rho)$ is a complete metric space.

If the income fluctuation problem \eqref{eq:IF} has a solution and the nonnegativity and borrowing constraints $0\le c_t\le a_t$ do not bind, the Euler equation implies
$$u'(c_t)=\E_t \beta_{t+1}R_{t+1}u'(c_{t+1}).$$
If $c_t=0$ or $c_t=a_t$, then clearly $u'(c_t)=u'(0)$ or $u'(c_t)=u'(a_t)$. Therefore combining these three cases, we can compactly express the Euler equation as
$$u'(c_t)=\min\set{\max\set{\E_t \beta_{t+1}R_{t+1}u'(c_{t+1}),u'(a_t)},u'(0)}.$$
Based on this observation, given a candidate consumption function $c\in \cC$, it is natural to update $c(a,z)$ by the value $\xi\in [0,a]$ that solves the Euler equation
\begin{equation}
u'(\xi)=\min\set{\max\set{\E_z \hat{\beta}\hat{R}u'(c(\hat{R}(a-\xi)+\hat{Y},\hat{Z})),u'(a)},u'(0)}.\label{eq:xi}
\end{equation}
The following lemma shows that such a $\xi$ uniquely exists.

\begin{lem}\label{lem:xi}
Suppose that $u'$ is continuous, positive, strictly decreasing, and $\E_{z,\hat{z}}\hat{\beta}\hat{R}<\infty$ and $\E_{z,\hat{z}}u'(\hat{Y})<\infty$ for all $(z,\hat{z})\in \ZZ^2$. Then for any $c\in \cC$, $a>0$, and $z\in \ZZ$, there exists a unique $\xi\in [0,a]$ satisfying \eqref{eq:xi}, with $\xi>0$ if $u'(0)=\infty$.
\end{lem}

When Assumptions~\ref{asmp:Inada}, \ref{asmp:spectral} hold and $c\in \cC$, $a>0$, and $z\in \ZZ$, by Lemma~\ref{lem:xi} we can define a unique number $Tc(a,z)\coloneqq \xi\in [0,a]$ that solves \eqref{eq:xi}. We call the operator $T$ defined on $\cC$ the \emph{time iteration operator}.\footnote{The time iteration operator was introduced by \cite*{Coleman1990}. Several papers such as \cite*{DattaMirmanReffett2002}, \cite*{Rabault2002}, \cite*{MorandReffett2003}, \cite*{Kuhn2013}, and \cite*{LiStachurski2014} use this approach to establish existence of solutions and study theoretical properties.} Analogous to Theorem 2.2 of \cite*{MaStachurskiToda2020JET}, we obtain the following existence and uniqueness result.

\begin{thm}\label{thm:exist}
Suppose Assumptions~\ref{asmp:Inada} and \ref{asmp:spectral} hold. Then $T$ is a monotone self map on $\cC$ and admits a unique fixed point $c\in \cC$, which is also the unique solution to the income fluctuation problem \eqref{eq:IF}. Furthermore, starting from any $c_0\in \cC$ and letting $c_n=T^nc_0$, we have $c_n\to c$.
\end{thm}

The proofs of Lemma~\ref{lem:xi} and Theorem~\ref{thm:exist} are relegated to Appendix~\ref{sec:IF}. Theorem~\ref{thm:exist} implies that the unique solution to the income fluctuation problem \eqref{eq:IF} can be computed by policy function iteration starting from any candidate consumption function $c_0\in \cC$; there are many such functions, for instance $c_0(a,z)=a$.

%See \citet*[Theorem 2.2]{MaStachurskiToda2020JET}.\footnote{In addition to Assumptions~\ref{asmp:Inada} and \ref{asmp:spectral}, \cite*{MaStachurskiToda2020JET} assume that the transition probability matrix $P$ is irreducible. However, irreducibility is required only for their ergodicity result, not for existence and uniqueness of a solution.}

%Exploiting policy function iteration, \cite*{MaStachurskiToda2020JET} show several properties such as
%\begin{enumerate*}[(i)]
%\item consumption and savings are increasing in wealth and
%\item consumption is increasing in income.
%\end{enumerate*}

\subsection{Asymptotic linearity of consumption functions}

To study the asymptotic behavior of consumption, we strengthen Assumption~\ref{asmp:Inada} as follows.

\begin{oneshot}[Assumption~\ref{asmp:Inada}']
The utility function exhibits constant relative risk aversion $\gamma>0$: we have
\begin{equation}
u(c)=\begin{cases}
\frac{c^{1-\gamma}}{1-\gamma}, & (\gamma\neq 1)\\
\log c. & (\gamma=1)
\end{cases}\label{eq:CRRA}
\end{equation}
Furthermore, letting $K$ be as in \eqref{eq:Ktheta}, the matrix $K(1-\gamma)$ is finite.% $\E_{z,\hat{z}}\hat{\beta}\hat{R}^{1-\gamma}<\infty$ for all $(z,\hat{z})\in \ZZ^2$.
\footnote{\label{fn:convention}We adopt the convention $\beta R^{1-\gamma}=(\beta R)R^{-\gamma}$ and $0\cdot \infty=0$. In particular, $\beta R^{1-\gamma}=0$ whenever $R=0$, even if $\beta>0$ and $\gamma>1$. %Then $\E_z \beta R^{1-\gamma}\in [0,\infty]$ is well-defined even if $\gamma>1$ and $(\beta,R)=(0,0)$ with positive probability.
This convention is necessary for avoiding tedious case-by-case analysis in the statements and proofs of theorems.}
\end{oneshot}

Theorem~\ref{thm:linear} below, which is our main theoretical result, shows that when the utility function exhibits constant relative risk aversion, the consumption functions are asymptotically linear and characterizes the asymptotic MPCs. To avoid overwhelming the reader with notation and technicalities, we maintain the additional condition that $K(1-\gamma)$ is finite as in Assumption~\ref{asmp:Inada}', although this condition can be dropped. Furthermore, Theorem~\ref{thm:linear} only provides a necessary and almost sufficient condition for the asymptotic MPCs to be zero. We provide a complete characterization in Theorem~\ref{thm:complete} below.

\begin{thm}[Asymptotic linearity]\label{thm:linear}
Suppose Assumptions~\ref{asmp:Inada}' and \ref{asmp:spectral} hold and let $K$ be as in \eqref{eq:Ktheta}. Then the following statements are true:
\begin{enumerate}
\item\label{item:linear1} If $r(K(1-\gamma))<1$, then for all $z\in \ZZ$ we have
\begin{equation}
\lim_{a\to\infty}\frac{c(a,z)}{a}\eqqcolon \bar{c}(z)>0,\label{eq:cbar}
\end{equation}
where $\bar{c}(z)=x^*(z)^{-1/\gamma}$ and $x^*=(x^*(z))_{z=1}^Z\in \R_+^Z$ is the unique finite solution to
the system of equations
\begin{equation}
x(z)=(Fx)(z)\coloneqq \left(1+(K(1-\gamma)x)(z)^{1/\gamma}\right)^\gamma, \quad z=1,\dots,Z. \label{eq:fixedpoint}
\end{equation}
\item\label{item:linear2} If $r(K(1-\gamma))\ge 1$ and $K(1-\gamma)$ is irreducible, then for all $z\in \ZZ$ we have
$$\lim_{a\to\infty}\frac{c(a,z)}{a}=0.$$
\end{enumerate}
\end{thm}

The proof of Theorem~\ref{thm:linear} is technical and relegated to Appendix~\ref{sec:proof}. Here we heuristically discuss the intuition for why we would expect the conclusion of Theorem~\ref{thm:linear} to hold. Suppose the limit \eqref{eq:cbar} exists. Assuming that the borrowing constraint does not bind, the Euler equation \eqref{eq:xi} implies
$$u'(\xi)=\E_z\hat{\beta}\hat{R}u'(c(\hat{R}(a-\xi)+\hat{Y},\hat{Z})),$$
where $\xi=c(a,z)$. Setting $u'(c)=c^{-\gamma}$ as in Assumption~\ref{asmp:Inada}', setting $c(a,z)=\bar{c}(z)a$ motivated by \eqref{eq:cbar}, multiplying both sides by $a^\gamma$, letting $a\to\infty$, and interchanging expectations and limits, it must be
\begin{equation}
\bar{c}(z)^{-\gamma}=\E_z\hat{\beta}\hat{R}^{1-\gamma}\bar{c}(\hat{Z})^{-\gamma}(1-\bar{c}(z))^{-\gamma}.\label{eq:Euler_asym}
\end{equation}
Multiplying both sides of \eqref{eq:Euler_asym} by $(1-\bar{c}(z))^\gamma$ and setting $x(z)=\bar{c}(z)^{-\gamma}$, after some algebra we obtain
\begin{equation}
x(z)=\left(1+\left(\E_z\hat{\beta}\hat{R}^{1-\gamma}x(\hat{Z})\right)^{1/\gamma}\right)^\gamma,\quad  z=1,\dots,Z.\label{eq:xzeq}
\end{equation}
Noting that $\hat{\beta},\hat{R}$ depend only on $Z$, $\hat{Z}$, and the \iid innovation $\hat{\zeta}$, we have
$$\E_z\hat{\beta}\hat{R}^{1-\gamma}x(\hat{Z})=\sum_{\hat{z}=1}^ZP_{z\hat{z}}\E_{z,\hat{z}}\hat{\beta}\hat{R}^{1-\gamma}x(\hat{z})=(K(1-\gamma)x)(z),$$
where we have used the definition of $K$ in \eqref{eq:Ktheta}. Therefore we can rewrite \eqref{eq:xzeq} as \eqref{eq:fixedpoint}. This discussion motivates the fixed point equation \eqref{eq:fixedpoint}.

Next, we discuss the intuition for the spectral condition $r(K(1-\gamma))\gtrless 1$. When the entries of the vector $x\in \R_+^Z$ are large, since $K\coloneqq K(1-\gamma)$ is a nonnegative matrix, it follows from the definition of $F$ in \eqref{eq:fixedpoint} that
$$Fx\approx Kx.$$
Since for large $x$ the function $x\mapsto Fx$ is almost linear, whether iterating $x\mapsto Fx$ converges or not depends on whether the largest eigenvalue of the coefficient matrix $K$ is less or greater than 1. When $r(K)<1$, $F$ in \eqref{eq:fixedpoint} behaves like a contraction and we would expect it to have a unique fixed point. When $r(K)\ge 1$, because $F$ is monotonic, we would expect the iteration of $x\mapsto Fx$ to diverge to infinity, and hence $\bar{c}(z)=x(z)^{-1/\gamma}$ to converge to 0.

Theorem~\ref{thm:linear} roughly says two things: with homothetic preferences,
\begin{enumerate*}[(i)]
\item consumption functions are asymptotically linear, and
\item the asymptotic MPCs can be zero.
\end{enumerate*}
The first point is not surprising based on the intuition of scale invariance with homothetic preferences, although we are not aware of a rigorous proof in a general setting.\footnote{Proposition 5 of \cite*{BenhabibBisinZhu2015} shows \eqref{eq:cbar} in the special case when $\beta<1$ is constant, $R,Y$ are \iid and mutually independent, have bounded supports in $(0,\infty)$, and satisfy $\E \beta R<1$ and $\E \beta R^{1-\gamma}<1$. \cite*{Carroll2020} provides a heuristic discussion similar to the one presented after Theorem~\ref{thm:linear} in the special case with constant $\beta<1$ and $R>0$. } The second point is nontrivial and surprising, and it depends on whether the condition
\begin{equation}
r(K(1-\gamma))<1 \label{eq:RIC}
\end{equation}
holds or not. A condition of the form $\E_z\hat{\beta}\hat{R}^{1-\gamma}<1$, which \cite*{Carroll2020} calls the ``return impatience condition'' and implies \eqref{eq:RIC}, is often required for the existence of a solution in dynamic programming problems with homothetic preferences.\footnote{See, for example, the discussion on p.~244 of \cite*{samuelson1969}, Assumption 1c of \cite*{alvarez-stokey1998}, Equation (9) of \cite*{krebs2006}, Equation (18) of \cite*{Toda2014JET}, Assumption 1(iii) of \cite*{BenhabibBisinZhu2015}, Equation (3) of \cite*{Toda2019JME}, or Equation (17) of \cite*{Carroll2020}.} The following proposition explains why this condition has often been assumed in the literature.

\begin{prop}\label{prop:noinc}
Suppose Assumption~\ref{asmp:Inada}' holds and $\gamma\neq 1$. Then the optimal consumption-saving problem \eqref{eq:IF} with zero income ($Y\equiv 0$) has a solution (with finite lifetime utility) if and only if \eqref{eq:RIC} holds. Under this condition, the optimal value and consumption functions are
\begin{subequations}\label{eq:noinc}
\begin{align}
V(a,z)&=\frac{x^*(z)}{1-\gamma}a^{1-\gamma},\label{eq:vnoinc}\\
c(a,z)&=x^*(z)^{-1/\gamma}a,\label{eq:cnoinc}
\end{align}
\end{subequations}
where $x^*\in \R_+^Z$ is the unique finite solution to \eqref{eq:fixedpoint}.
\end{prop}

Proposition~\ref{prop:noinc} implies that for a solution to the income fluctuation problem \eqref{eq:IF} to exist, the condition \eqref{eq:RIC} may be violated only if income $Y$ can be positive. In fact, the Inada condition $u'(0)=\infty$ for the CRRA utility and the condition $\E_{z,\hat{z}} u'(\hat{Y})<\infty$ in Assumption~\ref{asmp:spectral}\ref{item:EY} imply that $Y>0$ almost surely. Contrary to the intuition from the zero income model, Theorem~\ref{thm:exist} above shows that Assumptions~\ref{asmp:Inada} and \ref{asmp:spectral} are sufficient for the existence of a solution to general income fluctuation problems, and no conditions on risk aversion (including \eqref{eq:RIC}) are necessary.

As discussed above, Theorem~\ref{thm:linear} does not cover all possible cases as the matrix $K(1-\gamma)$ need not be finite or irreducible in particular applications. We can generalize Theorem~\ref{thm:linear} to cover all possible cases at the cost of making the notation slightly more complicated. To this end, let $K=K(1-\gamma)$ be as in \eqref{eq:Ktheta}, where each entry $K_{z\hat{z}}(1-\gamma)=P_{z\hat{z}}\E_{z,\hat{z}}\hat{\beta}\hat{R}^{1-\gamma}$ could be infinite (recall the convention in Footnote~\ref{fn:convention}). By relabeling the states $z=1,\dots,Z$ if necessary, without loss of generality we may assume that $K$ is block upper triangular,
\begin{equation}
K=\begin{bmatrix}
K_1 & \cdots & *\\
\vdots & \ddots & \vdots \\
0 & \cdots & K_J
\end{bmatrix},\label{eq:K}
\end{equation}
where each diagonal block $K_j$ is irreducible.\footnote{Recall that a square matrix $A$ is reducible if there exists a permutation matrix $P$ such that $P^\top AP$ is block upper triangular with at least two diagonal blocks. Matrices that are not reducible are called irreducible. Hence by induction a decomposition of the form \eqref{eq:K} is always possible. By definition scalars ($1\times 1$ matrices, including zero) are irreducible, so some $K_j$ in \eqref{eq:K} can be zero if it is $1\times 1$.} Partition $\ZZ$ as $\ZZ=\ZZ_1\cup \dots \cup \ZZ_J$ accordingly. Then we have the following complete characterization.

\begin{thm}[Complete characterization of asymptotic MPCs]\label{thm:complete}
Suppose Assumption~\ref{asmp:spectral} holds and the utility function exhibits constant relative risk aversion $\gamma>0$. Express $K=K(1-\gamma)$ as in \eqref{eq:K}. Define the sequence $\set{x_n}_{n=0}^\infty \in [0,\infty]^Z$ by $x_0=1$ and $x_n=Fx_{n-1}$, where $F$ is as in \eqref{eq:fixedpoint} and we apply the convention $0\cdot \infty=0$. Then $\set{x_n}$ monotonically converges to $x^*\in [1,\infty]^Z$, and the limit \eqref{eq:cbar} holds with $\bar{c}(z)=x^*(z)^{-1/\gamma}\in [0,1]$.

Furthermore, $\bar{c}(z)=0$ if and only if there exist $j$, $\hat{z}\in \ZZ_j$, and $m\in \N$ such that $K^m_{z\hat{z}}>0$ and $r(K_j)\ge 1$, where $r(K_j)=\infty$ if some entry of $K_j$ is infinite.
\end{thm}

\subsection{Implications of asymptotic linearity}

In this section we discuss the implications of our theoretical results.

As is clear from Theorems~\ref{thm:linear} and \ref{thm:complete}, the asymptotic MPCs $\bar{c}(z)$ depend only on the matrix $K(1-\gamma)$. Since the matrix $K$ in \eqref{eq:Ktheta} does not involve the income $Y$, we immediately obtain the following corollary. %, which depends on relative risk aversion $\gamma$, transition probability matrix $P$, ``multiplicative shocks'' $\beta$ and $R$, but not on ``additive shocks'' $Y$. 

\begin{cor}[Irrelevance of income]\label{cor:irrelevant}
Let everything be as in Theorem~\ref{thm:complete}. The asymptotic MPCs $\bar{c}(z)$ depend only on the relative risk aversion $\gamma$, transition probability matrix $P$, the discount factor $\beta$, and the return on wealth $R$, and not on income $Y$.
\end{cor}

Corollary~\ref{cor:irrelevant} verifies the intuition in \cite*{Gouin-BonenfantTodaParetoExtrapolation} that only ``multiplicative shocks'' such as $\beta$ and $R$ matter for characterizing the behavior of wealthy agents, and ``additive shocks'' such as $Y$ are irrelevant. They use the asymptotic MPCs to extrapolate the consumption functions and study the tail behavior of wealth in heterogeneous-agent models.

A natural question that arises from the discussion around \eqref{eq:RIC} is whether the case $r(K(1-\gamma))\ge 1$ (and hence zero asymptotic MPCs) is empirically plausible, or even theoretically possible. We argue in Section~\ref{sec:plausible} that $r(K(1-\gamma))\ge 1$ is empirically plausible. The following proposition shows that $\gamma>1$ is necessary for zero asymptotic MPCs. Furthermore, if persistent capital loss ($R(z,z,\zeta)<1$ with positive probability for some $z$ with $P_{zz}>0$) is possible, then zero asymptotic MPCs arise for sufficiently high risk aversion.

\begin{prop}\label{prop:gamless1}
If Assumption~\ref{asmp:spectral}\ref{item:rbeta} holds and $\gamma\le 1$, then $r(K(1-\gamma))<1$. If there exists $z\in \ZZ$ such that $P_{zz}>0$, $\beta(z,z,\zeta)>0$, and $0<R(z,z,\zeta)<1$ with positive probability, then $r(K(1-\gamma))\ge 1$ for sufficiently large $\gamma>1$.
\end{prop}

Example~\ref{exmp:lognormal} below (with \iid lognormal returns) shows that zero asymptotic MPCs are theoretically possible for any $\gamma>1$. The following proposition shows that the presence of capital income risk is crucial for zero asymptotic MPCs.

\begin{prop}\label{prop:Rconst}
Suppose Assumption~\ref{asmp:spectral}\ref{item:rbeta} holds and there is no capital income risk, so $R(z,\hat{z},\hat{\zeta})\equiv R$ is constant. If $r(K(1-\gamma))\ge 1$, then $R<1$.
\end{prop}

\begin{proof}
If $R(z,\hat{z},\hat{\zeta})=R$ is constant, then by \eqref{eq:Ktheta} we obtain $K(\theta)=R^\theta K(0)$. Therefore if Assumption~\ref{asmp:spectral}\ref{item:rbeta} holds and $r(K(1-\gamma))\ge 1$, then 
\begin{equation*}
1\le R^{1-\gamma}r(K(0))=R^{-\gamma}r(K(1))\implies R\le (r(K(1)))^{1/\gamma}<1.\qedhere
\end{equation*}
\end{proof}

With capital income risk, because capital loss is common, the second part of Proposition~\ref{prop:gamless1} states that zero asymptotic MPCs are possible. On the other hand, Proposition~\ref{prop:Rconst} implies that in a stationary environment with risk-free returns, zero asymptotic MPCs can arise only if the interest rate is negative, which is unrealistic.

\iffalse
However, the proposition does not rule out the possibility of zero asymptotic MPCs with positive interest rates in growing economies. For example, suppose that the income process is modified to
\begin{equation*}
Y_t=Y(Z_{t-1},Z_t,\zeta_t)\e^{gt},
\end{equation*}
where $g$ is the growth rate. Although our theory requires a stationary income process, it is straightforward to allow for constant growth in income by detrending the model when the utility function is CRRA. After simple algebra (\eg, Section 2.2 of \citealp*{Carroll2020}), instead of \eqref{eq:betaRY}, it suffices to use
\begin{subequations}\label{eq:betaRYtilde}
\begin{align}
\tilde{\beta}_t&=\beta(Z_{t-1},Z_t,\zeta_t) \e^{(1-\gamma)g},\\
\tilde{R}_t&=R(Z_{t-1},Z_t,\zeta_t)\e^{-g},\\
\tilde{Y}_t&=Y_t\e^{-gt}=Y(Z_{t-1},Z_t,\zeta_t),
\end{align}
\end{subequations}
which are stationary. Assuming constant $\beta$ and $R$ for simplicity, Assumption~\ref{asmp:spectral}\ref{item:rbeta} then becomes
\begin{align*}
1&>r(K(0))=\tilde{\beta}=\beta \e^{(1-\gamma)g},\\
1&>r(K(1))=\tilde{\beta}\tilde{R}=\beta R\e^{-\gamma g}.
\end{align*}
The condition for zero asymptotic MPC is then
\begin{equation*}
1\le r(K(1-\gamma))=\tilde{\beta}\tilde{R}^{1-\gamma}=\beta R^{1-\gamma}.
\end{equation*}
\fi

\subsection{Examples}

The system of fixed point equations \eqref{eq:fixedpoint} is in general nonlinear and does not admit a closed-form solution. Below, we discuss several examples with explicit solutions.

\begin{exmp}\label{exmp:gamma1}
If $\gamma=1$, then \eqref{eq:fixedpoint} becomes
$$x^*=1+K(0)x^*\iff x^*=(I-K(0))^{-1}1\gg 0.$$
Note that since $r(K(0))<1$ by Assumption~\ref{asmp:spectral}\ref{item:rbeta}, $(I-K(0))^{-1}=\sum_{k=0}^\infty K(0)^k$ exists and is nonnegative.
\end{exmp}

\begin{exmp}\label{exmp:z1}
If $b=b(z,\hat{z})=\E_{z,\hat{z}}\hat{\beta}\hat{R}^{1-\gamma}$ does not depend on $(z,\hat{z})$, then $K(1-\gamma)=bP$. If $x=k1$ is a multiple of the vector $1$, then $K(1-\gamma)x=bPk1=bk1$ because $P$ is a transition probability matrix. Thus if $b<1$, \eqref{eq:fixedpoint} reduces to
$$x^*(z)=(1+(bx^*(z))^{1/\gamma})^\gamma\iff x^*(z)=(1-b^{1/\gamma})^{-\gamma}\iff \bar{c}(z)=1-b^{1/\gamma}.$$
This example shows that with constant discounting ($\beta(z,\hat{z},\hat{\zeta})\equiv \beta$) and risk-free saving ($R(z,\hat{z},\hat{\zeta})\equiv R$), the asymptotic MPC is constant regardless of the income shocks: 
$$\bar{c}(z)=\begin{cases*}
1-(\beta R^{1-\gamma})^{1/\gamma} & if $\beta R^{1-\gamma}<1$,\\
0 & otherwise.
\end{cases*}$$
This case has been studied in \cite*{Carroll2020} in an \iid setting.
\end{exmp}

\begin{exmp}\label{exmp:lognormal}
Suppose the return on wealth $R_t=R(Z_{t-1},Z_t,\zeta_t)$ does not depend on $(Z_{t-1},Z_t)$, so $R_t=R(\zeta_t)$. Assume further that $\log R_t$ is normally distributed with standard deviation $\sigma$ and mean $\mu-\sigma^2/2$, so $\E R=\e^{\mu}$. Let the discount factor $\beta=\e^{-\delta}$ be constant, where $\delta>0$ is the discount rate. Then using the property of the normal distribution, we obtain
\begin{align*}
&1>\E \beta R=\e^{-\delta+\mu}\iff \delta>\mu,\\
&1>\E \beta R^{1-\gamma}=\e^{-\delta+(1-\gamma)(\mu-\gamma\sigma^2/2)}\iff \delta>(1-\gamma)\left(\mu-\frac{1}{2}\gamma\sigma^2\right).
\end{align*}
Therefore assuming $\delta>\mu$ for Assumption~\ref{asmp:spectral}\ref{item:rbeta} to hold, it follows from Example~\ref{exmp:z1} that
$$\bar{c}(z)=\begin{cases*}
1-\e^{-\psi\delta-(1-\psi)(\mu-\gamma\sigma^2/2)}>0 & if $\delta>(1-\gamma)\left(\mu-\frac{1}{2}\gamma\sigma^2\right)$,\\
0 & otherwise,
\end{cases*}$$
where $\psi=1/\gamma$ is the elasticity of intertemporal substitution. If $\gamma>1$, then $(1-\gamma)(\mu-\gamma\sigma^2/2)\to\infty$ as $\gamma,\sigma\to\infty$, so the asymptotic MPC is 0 if risk aversion or volatility is sufficiently high.
\end{exmp}

\subsection{Asymptotic MPCs and saving rates}

In this section we apply our theory of asymptotic MPCs to shed light on the saving rate of the rich.

As is common in the literature, we define an agent's saving rate by the change in net worth divided by total income excluding capital loss (to prevent the denominator from becoming negative):
\begin{equation}
s_{t+1}=\frac{a_{t+1}-a_t}{\max\set{(R_{t+1}-1)(a_t-c_t),0}+Y_{t+1}}.\label{eq:saverate1}
\end{equation}
For $x\in \R$, define its positive and negative parts by $x^+=\max\set{x,0}$ and $x^-=-\min\set{x,0}$. Then $x=x^+-x^-$. Using the budget constraint \eqref{eq:IF_budget}, the saving rate \eqref{eq:saverate1} can be rewritten as
\begin{align}
s_{t+1}&=\frac{[(R_{t+1}-1)^+-(R_{t+1}-1)^-](a_t-c_t)+Y_{t+1}-c_t}{(R_{t+1}-1)^+(a_t-c_t)+Y_{t+1}}\notag \\
&=1-\frac{(\hat{R}-1)^-(1-c/a)+c/a}{(\hat{R}-1)^+(1-c/a)+\hat{Y}/a}\in (-\infty,1).\label{eq:saverate2}
\end{align}
Letting $a\to\infty$, the saving rate of an infinitely wealthy agent becomes
\begin{equation}
\bar{s}\coloneqq 1-\frac{(\hat{R}-1)^-(1-\bar{c})+\bar{c}}{(\hat{R}-1)^+(1-\bar{c})} \in [-\infty,1],\label{eq:saverateinf}
\end{equation}
where $\bar{c}$ is the asymptotic MPC. Under what conditions can the saving rate \eqref{eq:saverate2} be increasing in wealth, and in particular, can the asymptotic saving rate \eqref{eq:saverateinf} become positive? The following proposition provides a negative answer within a class of models.

\begin{prop}\label{prop:Bewley}
Consider a canonical \cite*{bewley1977} model in which agents are infinitely-lived and relative risk aversion $\gamma$, discount factor $\beta$, and return on wealth $R$ are constant. Then in the stationary equilibrium the asymptotic saving rate \eqref{eq:saverateinf} is negative.
\end{prop}

%\citet*[Figure 1]{DeNardiFella2017} numerically show that the saving rate in a particular Bewley model is decreasing in wealth and becomes negative as agents become wealthier. 
Proposition~\ref{prop:Bewley} proves that the negativity of the asymptotic saving rate is inevitable in any canonical (stationary) Bewley model.\footnote{This result has a similar flavor to \cite*{StachurskiToda2019JET}, who prove that canonical Bewley models cannot explain the tail behavior of wealth.} Thus, these models are unable to explain the observed positive saving rates of the rich. The following proposition shows that just by allowing $\beta$ or $R$ to be stochastic need not solve the problem when $\bar{c}>0$.

\begin{prop}\label{prop:Bewley2}
Consider a \cite*{bewley1977} model in which agents are infinitely-lived, relative risk aversion $\gamma$ is constant, and $\set{\beta_t,R_t}_{t\ge 1}$ is \iid with $\E \beta R^{1-\gamma}<1$. If the stationary equilibrium wealth distribution has an unbounded support, then the asymptotic saving rate \eqref{eq:saverateinf} evaluated at $\hat{R}=\E R$ is nonpositive.
\end{prop}

One possible explanation for the positive and increasing saving rates is to consider models with discount factor or return heterogeneity. If $r(K(1-\gamma))\ge 1$, then by Theorem~\ref{thm:linear} we have $\bar{c}=0$ and hence the asymptotic saving rate becomes $\bar{s}=1>0$ using \eqref{eq:saverateinf}.\footnote{Another possibility is to consider overlapping generations models. \citet*[Theorem 9]{StachurskiToda2019JET} present a model with random birth/death and show that it is possible to have $\beta R>1$ in equilibrium. In this case, by the proof of Proposition~\ref{prop:Bewley}, we have $\bar{s}>0$.}

\section{Empirical plausibility of zero asymptotic MPCs}\label{sec:plausible}

So far we have theoretically characterized the asymptotic MPCs in Theorems~\ref{thm:linear} and \ref{thm:complete}, and showed in Proposition~\ref{prop:gamless1} that zero asymptotic MPCs arise whenever capital loss is possible and risk aversion is sufficiently high. The remaining issue is whether zero asymptotic MPCs (and hence asymptotic saving rates equal to 1) can arise in empirically plausible settings. To address this issue, in this section we provide an empirically plausible proof of concept in a partial equilibrium setting to study the saving rate of the rich and wealth inequality. Analyzing a fully calibrated general equilibrium model that matches various aspects of the data is left for applied researchers.

%To address this issue, in this section we consider a stylized numerical example calibrated from U.S.\ data. We view this exercise as a proof of concept: solving a fully specified general equilibrium model that matches various aspects of the data is beyond the scope of the paper.

\subsection{Model and calibration}
\paragraph{Model}
The economy is populated by a continuum of ex ante identical, infinitely-lived dynastic households with CRRA utility with constant discount factor $\beta>0$ and relative risk aversion $\gamma>0$. A typical agent (head of household) can be in one of the following states: employed worker ($z=1$), unemployed worker ($z=2$), and entrepreneur ($z=3$), so the state space is $\ZZ=\set{1,2,3}$. The state process $\set{Z_t}_{t=0}^\infty$ is independent across households and evolves as a Markov chain with transition probability matrix $P$.

Letting $Z_t$ be the time $t$ state of a typical agent, we suppose that labor income is $Y_t=Y(Z_t)\e^{gt}>0$, where $Y:\ZZ\to (0,\infty)$ and $g$ is the aggregate growth rate of the economy. As for the return on wealth, workers (employed and unemployed) save only at gross risk-free rate $R_f>0$, whereas entrepreneurs enjoy excess returns as follows. Let $X$ be the gross excess return on risky investment, so the gross return on investment is $R_fX$. Entrepreneurs invest fraction $\theta$ of their wealth into the risky asset and are subject to capital income tax at rate $\tau_k$ that applies to excess returns. Therefore the return on wealth of a typical entrepreneur is
\begin{equation*}
R_f(1+(1-\tau_k)(X_t-1)\theta),
\end{equation*}
where for simplicity we assume that $\set{X_t}_{t=0}^\infty$ is \iid across agents and time. Finally, to introduce social mobility, we suppose that the head of a household dies with probability $p$ each period and the heir inherits the financial wealth after paying the estate tax at rate $\tau_e$. In summary, we can write the return on wealth as
\begin{equation*}
R(Z_{t-1},Z_t,\zeta_t)=\begin{cases}
(1-\tau_e d_t)R_f, & (Z_{t-1}=1,2)\\
(1-\tau_e d_t)R_f(1+(1-\tau_k)(X_t-1)\theta), & (Z_{t-1}=3)
\end{cases}
\end{equation*}
where $d_t$ is the indicator function of death (so $d_t=1$ if the household head dies and $d_t=0$ otherwise) and the \iid shock is denoted by $\zeta_t=(X_t,d_t)$.

Although the theoretical results in Section~\ref{sec:AL} requires a stationary income process, it is straightforward to allow for constant growth in income by detrending the model when the utility function is CRRA. After simple algebra (\eg, Section 2.2 of \citealp*{Carroll2020}), instead of \eqref{eq:betaRY}, it suffices to use
\begin{subequations}\label{eq:betaRYtilde}
\begin{align}
\tilde{\beta}_t&=\beta(Z_{t-1},Z_t,\zeta_t) \e^{(1-\gamma)g}=\beta \e^{(1-\gamma)g},\\
\tilde{R}_t&=R(Z_{t-1},Z_t,\zeta_t)\e^{-g},\\
\tilde{Y}_t&=Y_t\e^{-gt}=Y(Z_{t-1},Z_t,\zeta_t)=Y(Z_t),
\end{align}
\end{subequations}
which are stationary.

\paragraph{Asset returns}
We model one period as a month. To calibrate the asset return parameters, we use the 1947--2018 monthly data for U.S.\ stock market returns (volume-weighted index including dividends) and risk-free rates from the updated spreadsheet of \cite*{welch-goyal2008}.\footnote{\url{http://www.hec.unil.ch/agoyal/docs/PredictorData2018.xlsx}.} Their spreadsheet contains monthly nominal stock and risk-free returns as well as the inflation. From these we construct the real gross stock and risk-free returns $R_t^s,R_t^f$. We estimate the log risk-free rate as $\log R_f=\E[\log R_t^f]=5.3477\times 10^{-4}$ (annual rate 0.65\%). We suppose that gross excess return $X_t$ is lognormal ($\log X\sim N(\mu,\sigma^2))$ and estimate $\mu=5.4079\times 10^{-3}$ and $\sigma=0.0414$ from the mean and standard deviation of the log excess returns $\log R_t^s-\log R_t^f$. For computational purposes, we discretize the distribution of $\log X$ using the 7-point Gauss-Hermite quadrature.

\paragraph{Portfolio}
To calibrate the risky portfolio share $\theta$, we use the 1913--2012 wealth share data of the wealthiest households in U.S.\ estimated by \cite*{saez-zucman2016}. Specifically, in Table B5b of their Online Appendix, they report the composition of wealth of the top 0.01\% across asset groups (equities, fixed income claims, housing, business assets, and pensions). We classify equities, business, and pension as ``risky asset'' and fixed income claims and housing as ``risk-free asset'' to compute the portfolio share $\theta$ for all years,\footnote{These portfolio shares are relatively stable over time. Although the classification of housing and pension may be ambiguous, because these two categories comprise a small fraction (about 10\%) of the portfolio, choosing different classifications yields quantitatively similar results.} take the average across all years, and obtain $\theta=0.6373$.

\paragraph{Income process}

We choose the transition probability matrix $P$ such that
\begin{enumerate*}[(i)]
\item conditional on remaining a worker, unemployment lasts on average for 3 months,
\item conditional on being a worker, unemployment rate is 5\%,
\item an entrepreneur becomes a worker at annual rate 2\%,\footnote{\citet*[Table 1]{GilchristYankovZakrajsek2009} document that the credit spread of large firms is 192 basis points, or about 2\%. We interpret firm exit as switching from entrepreneur to worker.} and
\item the fraction of entrepreneurs is 11.5\% (fraction of ``active business owners'' in \citealp*[Table 1]{CagettiDeNardi2006}).
\end{enumerate*}
The resulting transition probability matrix is
\begin{equation*}
P=\begin{bmatrix}
0.9822 & 0.0175 & 0.0002\\
0.3333 & 0.6665 & 0.0002\\
0.0016 & 0.0001 & 0.9983
\end{bmatrix}.
\end{equation*}
We set $(Y(1),Y(2),Y(3))=(1,0.2,2.5)$ so that the income of an unemployed worker is 20\% of an employed worker, and an entrepreneur earns 2.5 times as much as an employed worker.

\paragraph{Other parameters}

We calibrate the remaining parameters as follows. The discount factor is $\beta = \e^{-\delta/12}$ with $\delta=0.04$ so that the annual discounting is $4\%$. The death probability of the household head is $p=\e^{-1/(25\times 12)}$ so that a generation lasts for 25 years on average. The capital income tax rate is $\tau_k=0.25$ based on the estimate in \cite*{mcdaniel2007average} using national account statistics. The estate tax rate is $\tau_e=0.4$, which is the current value in U.S. We calibrate the growth rate $g$ from the U.S.\ real per capita GDP in 1947--2018 and obtain $g=1.6208\times 10^{-3}$ at the monthly frequency.

\subsection{Empirical plausibility of zero asymptotic MPCs}

\paragraph{Asymptotic MPCs}

In the current setting, Assumption~\ref{asmp:Inada}' and conditions~\ref{item:Ebeta} and \ref{item:EY} of Assumption~\ref{asmp:spectral} obviously hold. To apply Theorems~\ref{thm:exist} and \ref{thm:linear}, it remains to verify $r(K(0))<1$, $r(K(1))<1$, and determine whether $r(K(1-\gamma)) \gtrless 1$, where we compute $K$ in \eqref{eq:Ktheta} using the effective discount factors and returns in \eqref{eq:betaRYtilde}. Figure~\ref{fig:rK} shows the determination of the asymptotic MPC $\bar{c}(z)$ when we change the relative risk aversion $\gamma$ and the annual discount rate $\delta$. The blue dashed and dotted lines show the boundaries of the existence conditions $r(K(0))<1$ and $r(K(1))<1$ in Assumption~\ref{asmp:spectral}\ref{item:rbeta}. By Theorem~\ref{thm:exist}, for any $(\gamma,\delta)$ configuration above these lines, a solution to the income fluctuation problem exists. The red curve shows the discount rate corresponding to $r(K(1-\gamma))=1$. By Theorem~\ref{thm:linear}, for any $(\gamma,\delta)$ configuration above (below) this curve, we obtain $\bar{c}(z)>0$ ($\bar{c}(z)=0$). Figure~\ref{fig:rK} reveals that the asymptotic MPCs can be zero if relative risk aversion is moderately high (above 3).

\begin{figure}[!htb]
\centering
\includegraphics[width=0.7\linewidth]{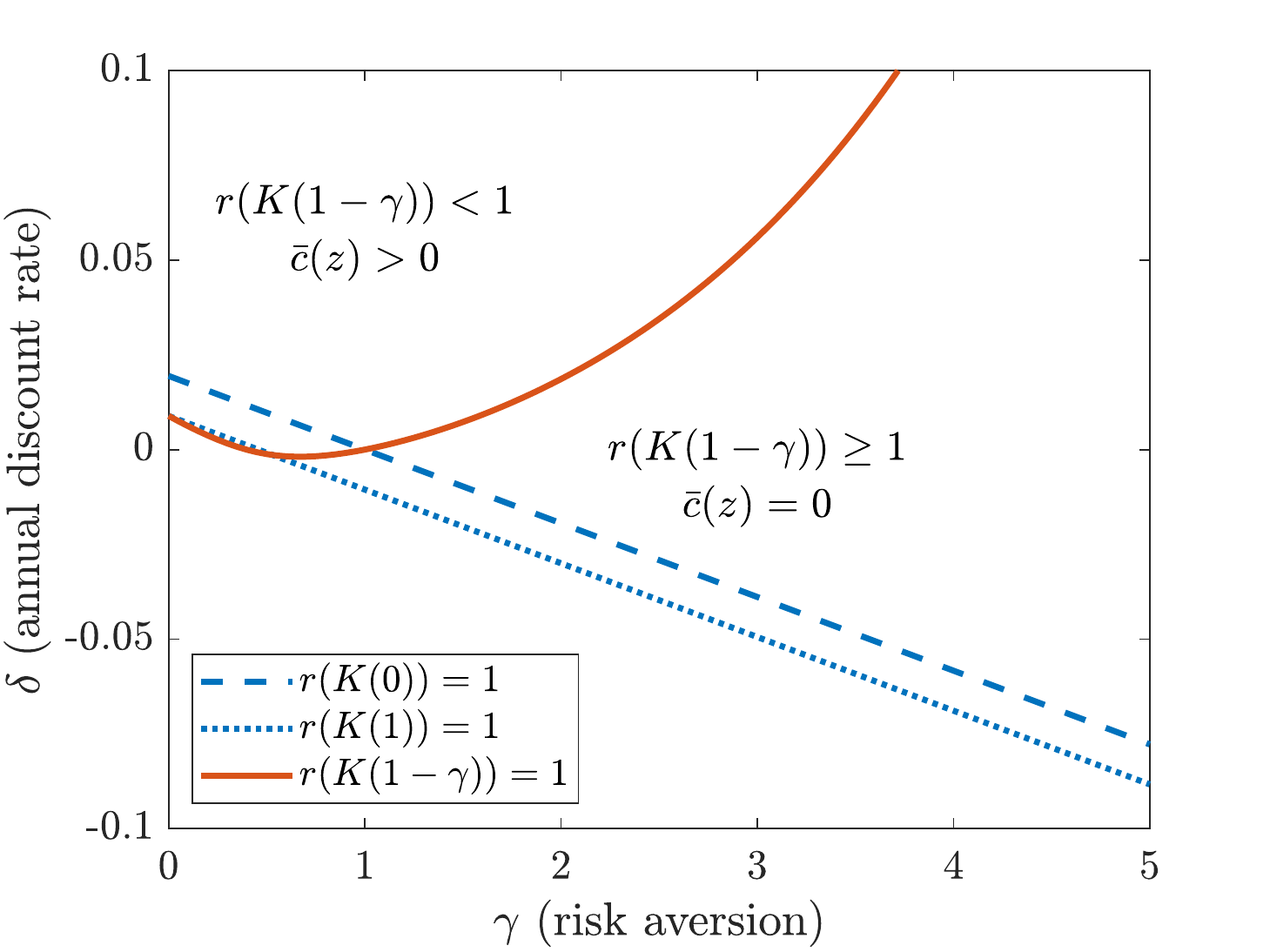}
\caption{Determination of asymptotic MPCs.}\label{fig:rK}
\end{figure}

\paragraph{Consumption functions}

We next solve the model for $\gamma=2,4$ using policy function iteration.\footnote{To avoid the root-finding in \eqref{eq:xi} and speed up the algorithm, we use the endogenous grid point method \citep*{Carroll2006}.} According to Figure~\ref{fig:rK} and Theorem~\ref{thm:exist}, a unique solution exists in each case given the annual discount rate $\delta=0.04$. Figure~\ref{fig:cf} shows the optimal consumption rule. Consistent with our theory, for $\gamma=2$ ($\bar{c}(z)>0$) the consumption functions are approximately linear with positive slopes for high asset level. When $\gamma=4$ ($\bar{c}(z)=0$), the consumption functions show a distinctive concave pattern.

\begin{figure}[!htb]
\centering

\begin{subfigure}{0.48\linewidth}
\includegraphics[width=\linewidth]{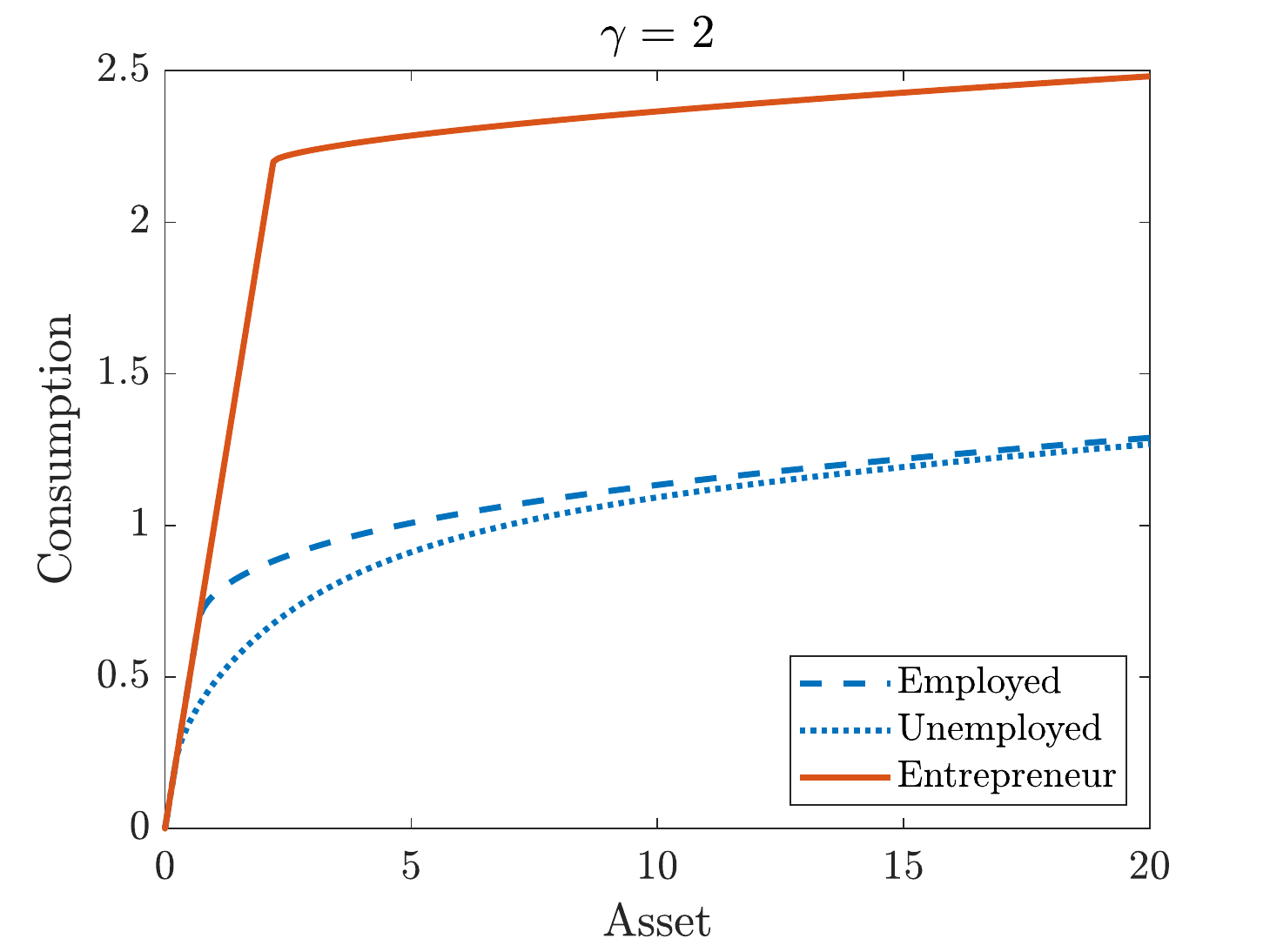}
\end{subfigure}
\begin{subfigure}{0.48\linewidth}
\includegraphics[width=\linewidth]{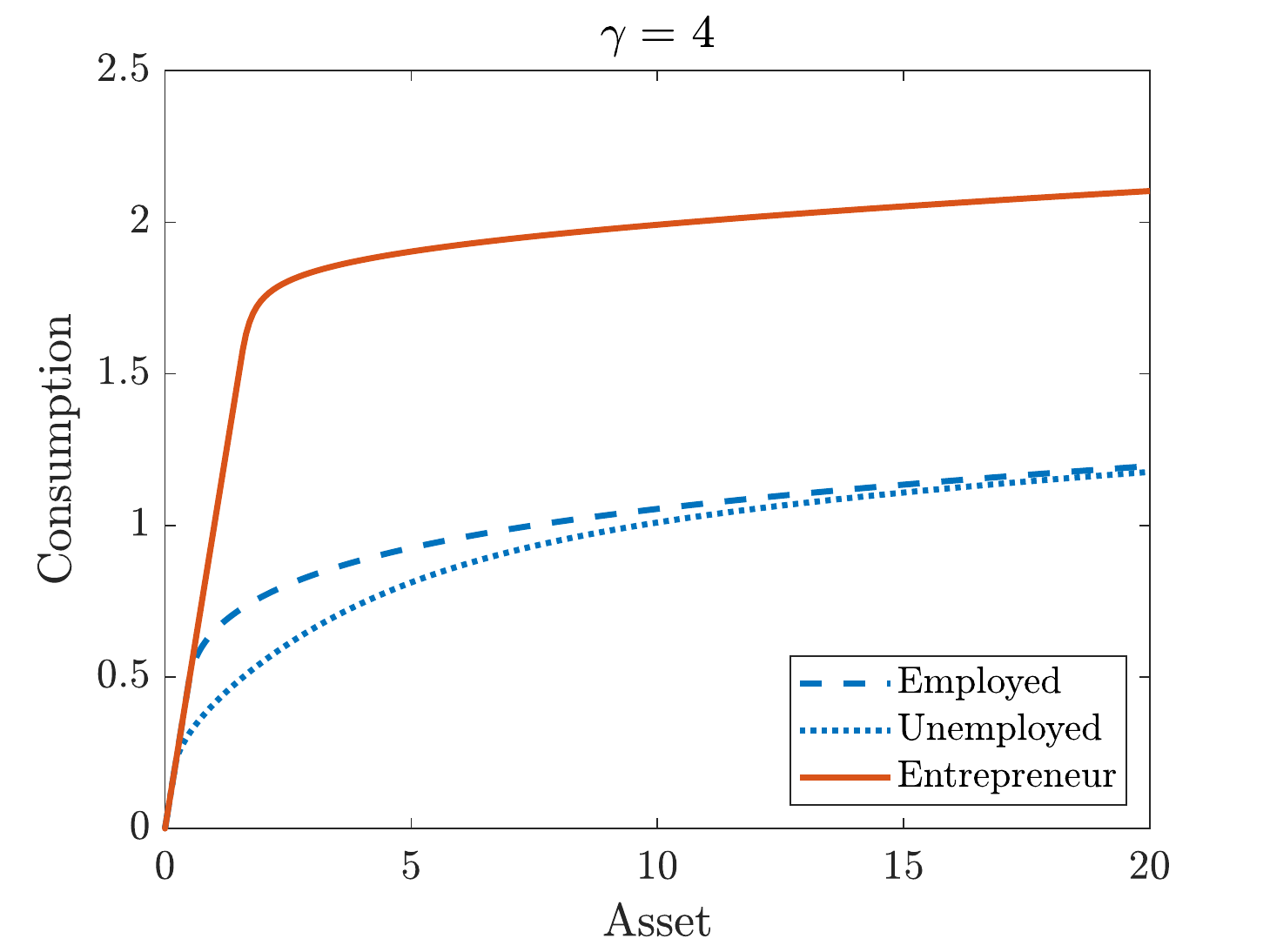}
\end{subfigure}
\begin{subfigure}{0.48\linewidth}
\includegraphics[width=\linewidth]{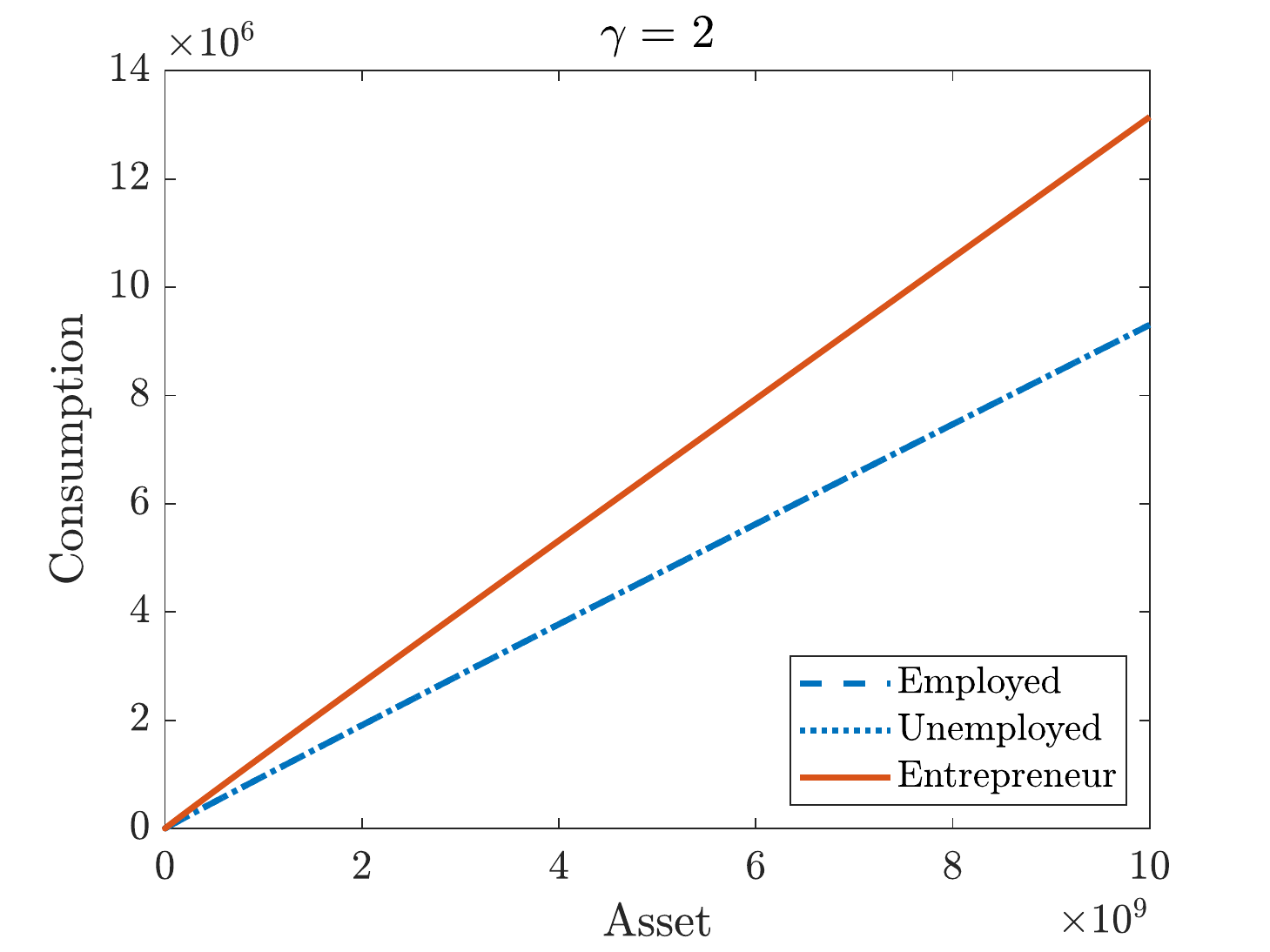}
\end{subfigure}
\begin{subfigure}{0.48\linewidth}
\includegraphics[width=\linewidth]{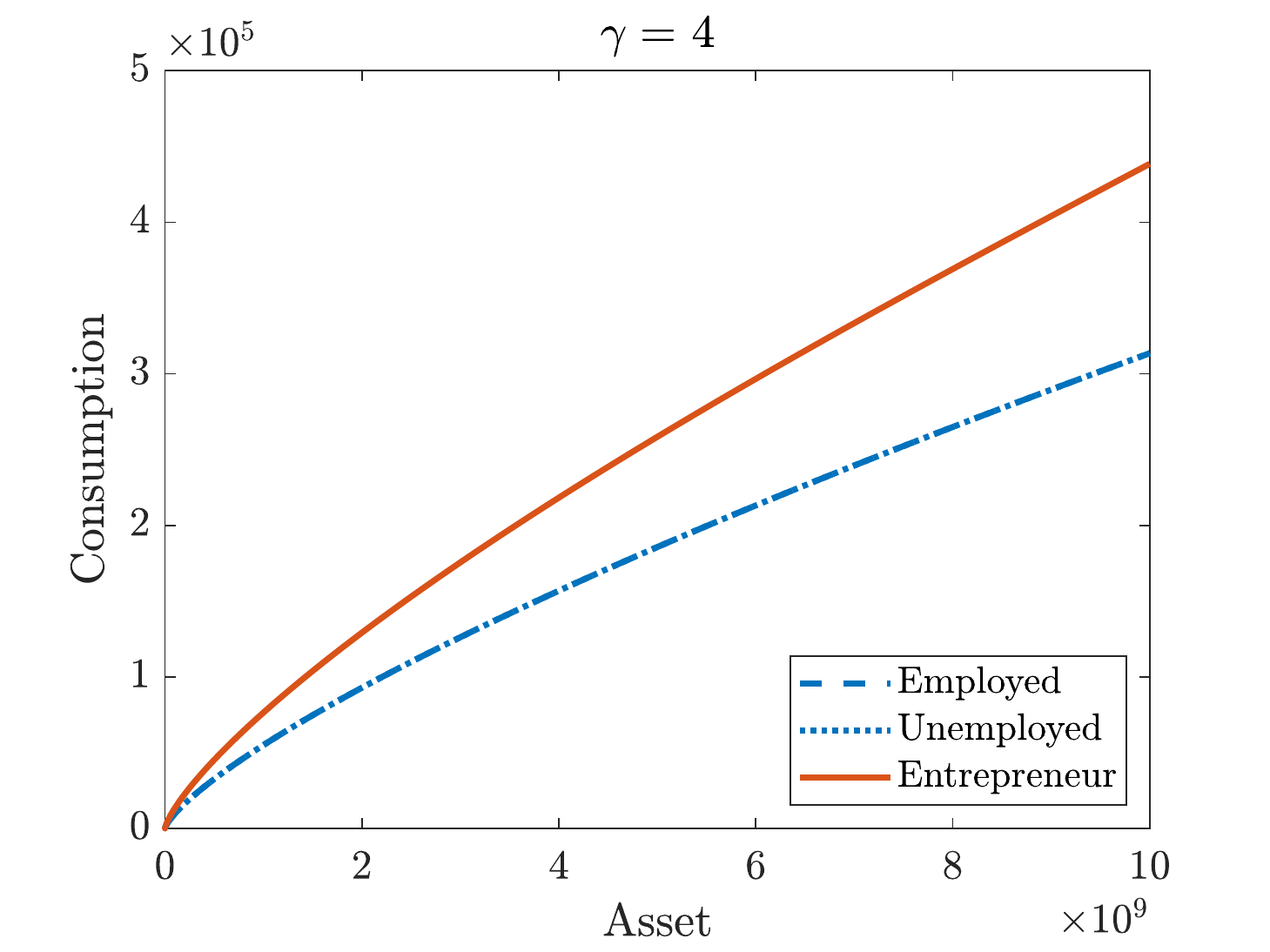}
\end{subfigure}
\caption{Consumption functions.}\label{fig:cf}
\caption*{\footnotesize Note: The top and bottom panels plot the consumption functions in the range $a\in [0,20]$ and $a\in [0,10^{10}]$, respectively. Here and in other figures, the left (right) panels correspond to $\gamma=2$ ($\gamma=4$).}
\end{figure}

Figure~\ref{fig:cr} plots the consumption rates ($c(a,z)/a$) in log-log scale. We see that the consumption rates are decreasing in wealth for each state. For $\gamma=2$, as asset level gets large, the asymptotic MPCs approach to positive constants that coincide with the theoretical values calculated based on Theorem~\ref{thm:linear} (dotted lines). Thus the consumption functions are asymptotically linear, consistent with the theorem. For $\gamma=4$, the consumption rates exhibit a clear decreasing trend even when asset is extremely large ($a \approx 10^{10}$), which is consistent with zero asymptotic MPC established in Theorem~\ref{thm:linear}.

\begin{figure}[!htb]
\centering
\begin{subfigure}{0.48\linewidth}
\includegraphics[width=\linewidth]{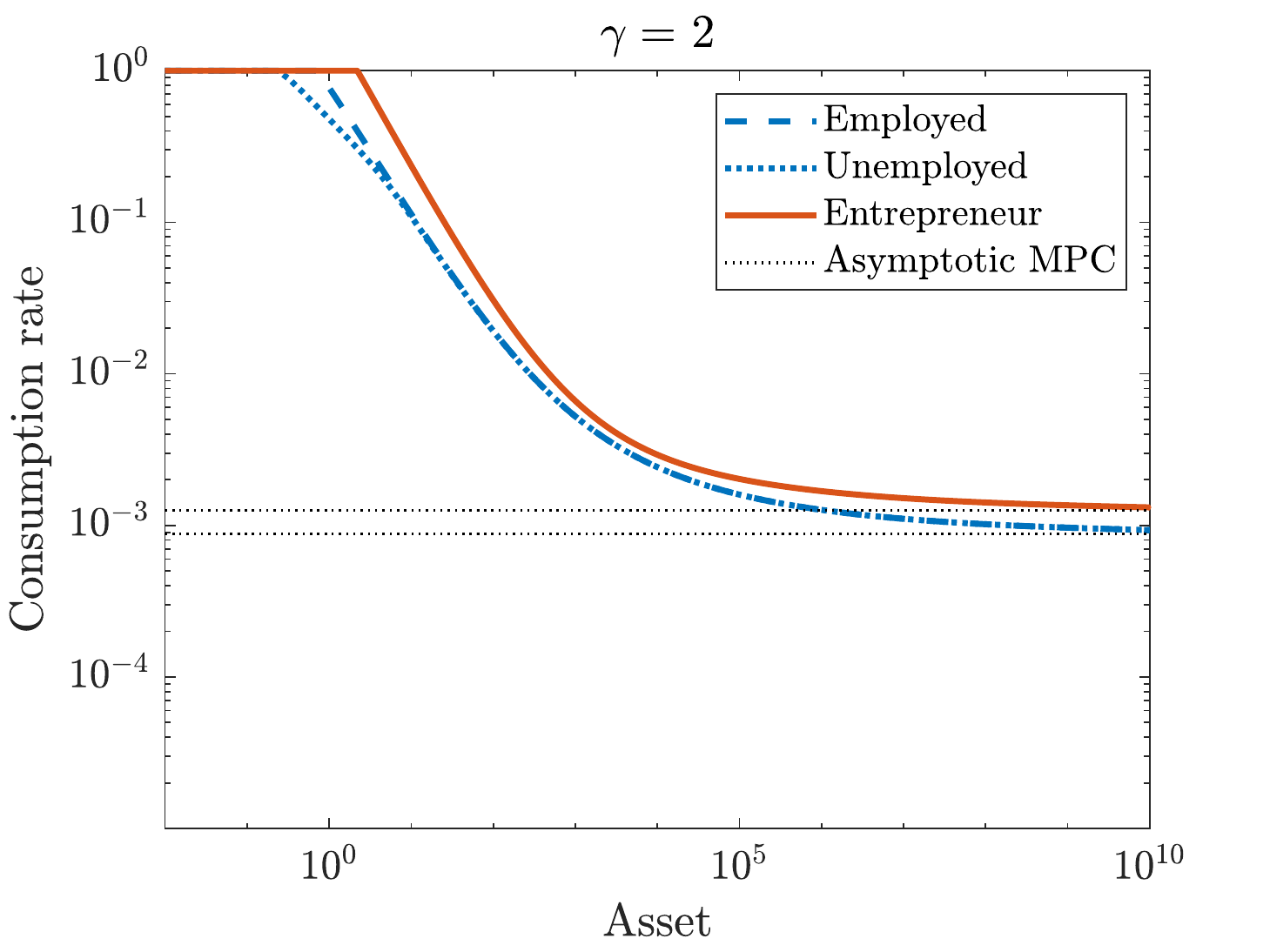}
\end{subfigure}
\begin{subfigure}{0.48\linewidth}
\includegraphics[width=\linewidth]{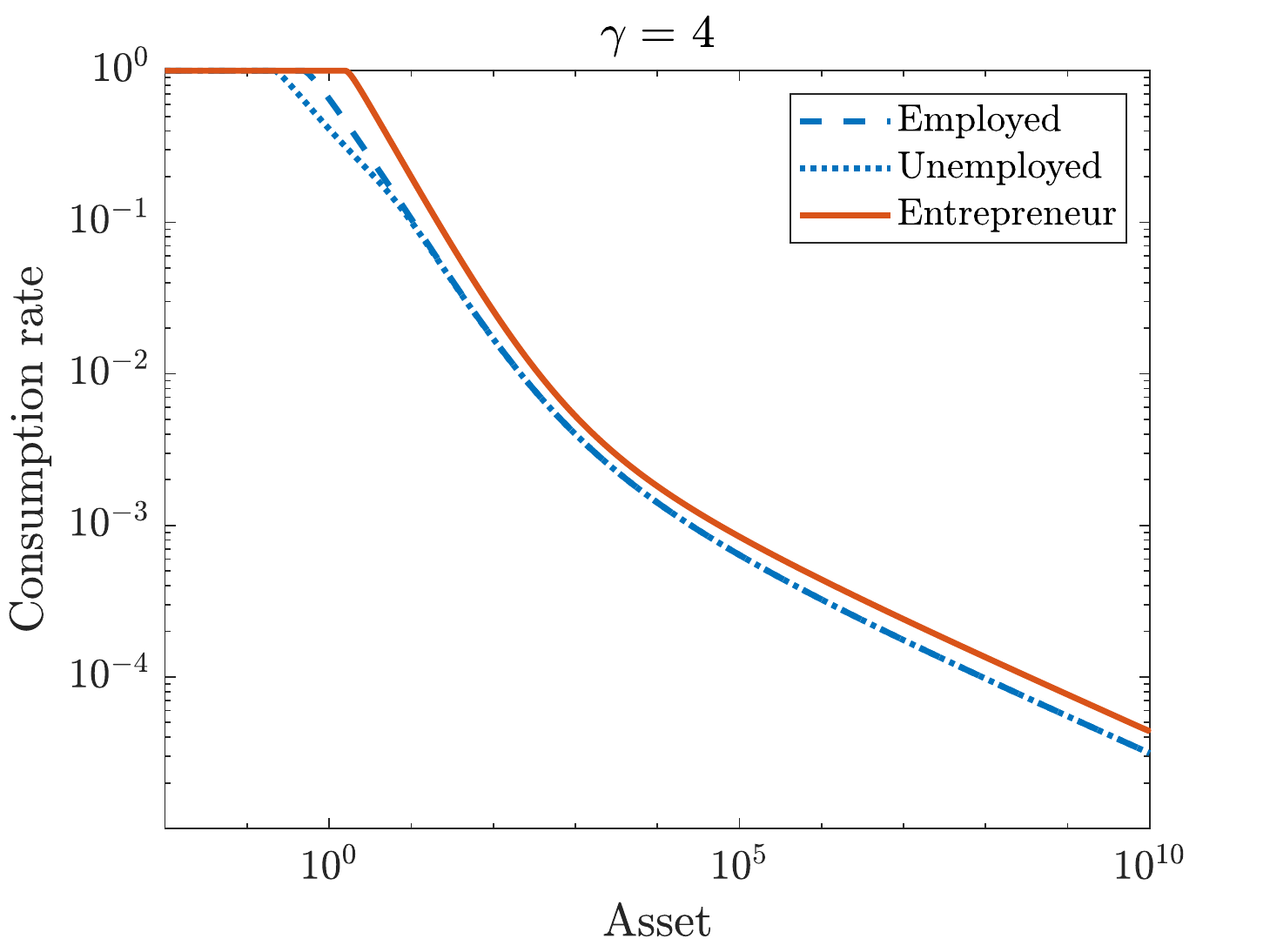}
\end{subfigure}
\caption{Consumption rates.}\label{fig:cr}
\end{figure}

\paragraph{Saving rates}

We compute the saving rate in each state using the definition \eqref{eq:saverate1}. In our setting, $s_t$ depends on $(Z_{t-1},Z_t,\zeta_t)$, which can take $3\times 3\times 7\times 2=126$ states (3 states for $Z_{t-1}$ and $Z_t$ each, 7 states for the discretized gross excess return $X_t$, and 2 states for indicator of death $d_t$). To reduce the dimension, Figure~\ref{fig:sr} shows the saving rates assuming $Z_{t-1}=Z_t=z$ (no change in occupation), $\log X_t=\mu$ (median return), and $d_t=0$ (survival).

When wealth is low, the borrowing constraint binds and labor income is the only source of income and net worth accumulation. Using \eqref{eq:saverate2}, we obtain
\begin{equation*}
s_{t+1}=\frac{\hat{Y}-a}{\hat{Y}}=1-a/\hat{Y},
\end{equation*}
which is decreasing in asset. When $\hat{R}>1$, by \eqref{eq:saverate2} we obtain
\begin{equation*}
s_{t+1}=1-\frac{c}{(\hat{R}-1)(a-c)+\hat{Y}}.
\end{equation*}
Thus when wealth is moderately high so that $c<a$ but $(\hat{R}-1)(a-c)\ll \hat{Y}$, the saving rate is decreasing because $c$ is increasing in $a$ but the denominator is roughly constant at $\hat{Y}>0$. The saving rate starts to increase when wealth is relatively high ($\approx 100\sim 1000$). When $\gamma=2$, the saving rate of extremely wealthy entrepreneurs is positive. This finding does not contradict Proposition~\ref{prop:Bewley2} because the model features Markovian shocks. However, for the relevant region of the state space (say $a\le 10^4$), where agents spend most time, the saving rate is either small or negative. On the other hand, when $\gamma=4$, the saving rate of entrepreneurs remains large and positive, and the asymptotic saving rate equals 1. This example illustrates that the empirically observed large positive and increasing saving rate (see Figure~1 of \cite*{FagerengHolmMollNatvikWP}) could potentially be explained by models with capital income risk, particularly those with zero asymptotic MPCs.

\begin{figure}[!htb]
\centering
\begin{subfigure}{0.48\linewidth}
\includegraphics[width=\linewidth]{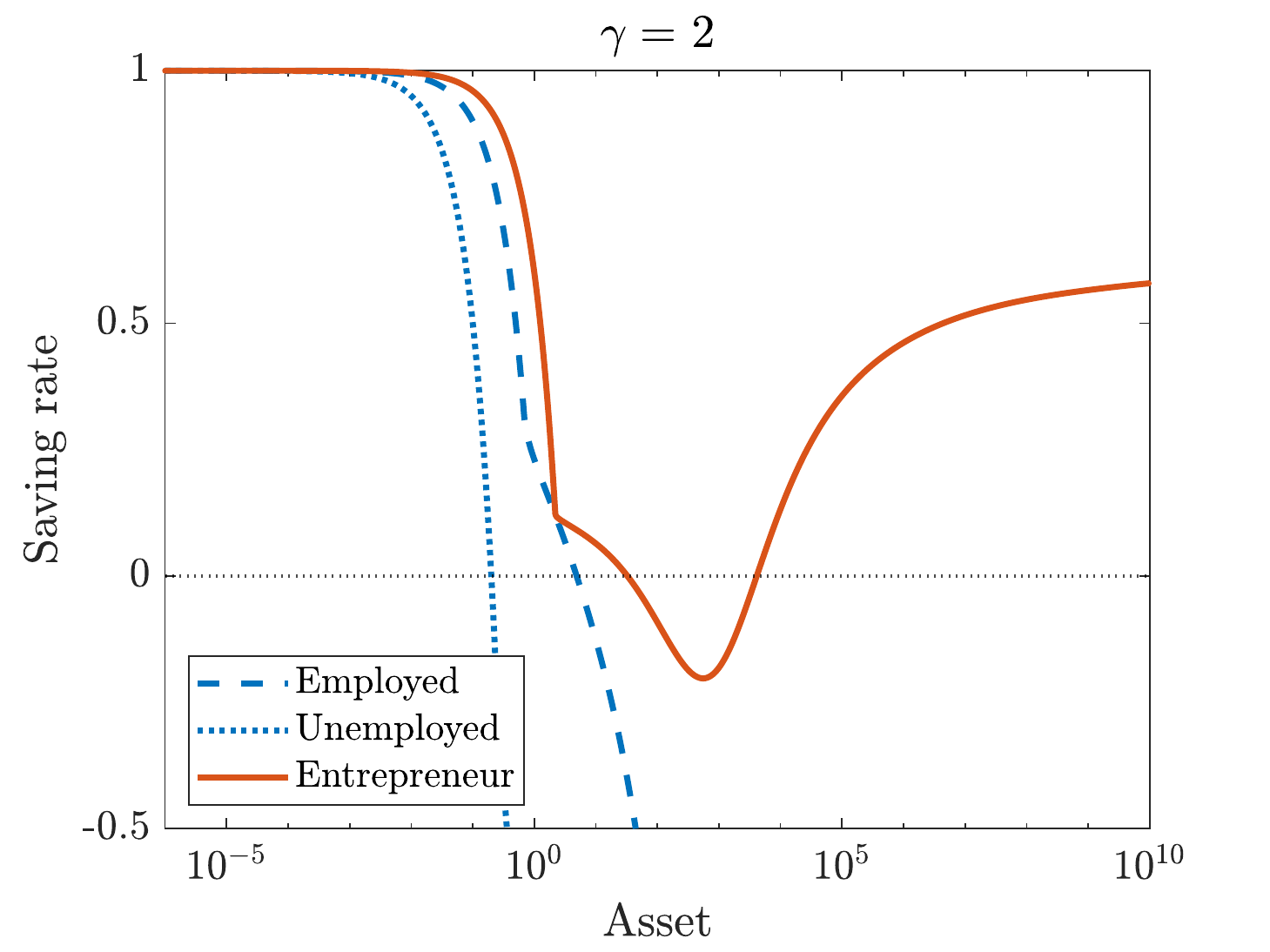}
\end{subfigure}
\begin{subfigure}{0.48\linewidth}
\includegraphics[width=\linewidth]{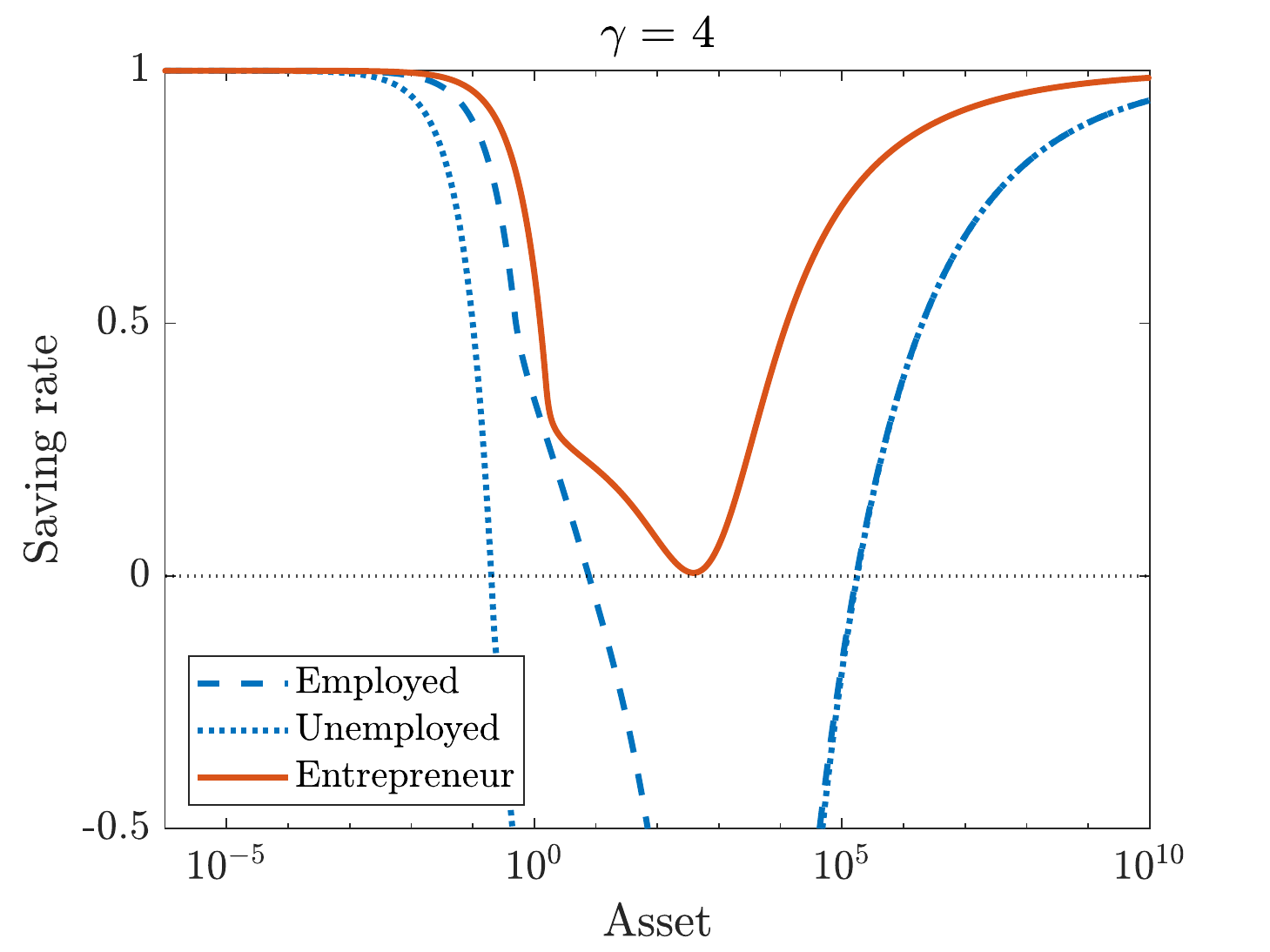}
\end{subfigure}
\caption{Saving rates.}\label{fig:sr}
\end{figure}

\subsection{Zero asymptotic MPCs and wealth inequality}

Finally, we investigate the implication of saving rates on the stationary wealth distribution. Let $G=(G_{z\hat{z}})$ be the matrix whose $(z,\hat{z})$ entry is the conditional expected return $\E R(z,\hat{z},\hat{\zeta})$. Then Theorems 3.1 and 3.2 of \cite*{MaStachurskiToda2020JET} imply that a sufficient condition for the existence of a unique stationary wealth distribution is
\begin{equation}
r(P\odot G)<1,\label{eq:ergodicCond}
\end{equation}
where $\odot$ denotes the Hadamard (entry-wise) product of matrices. In our numerical example, we have $r(P\odot G)=0.9991<1$, so \eqref{eq:ergodicCond} holds.

Due to the presence of capital income risk, the wealth distribution has a Pareto upper tail as shown by \citet*[Theorem 3.3]{MaStachurskiToda2020JET}. Because the wealth distribution has a heavy upper tail, truncating the distribution using a finite grid leads to substantial truncation error. Therefore to numerically compute the stationary wealth distribution, we apply the Pareto extrapolation method of \cite*{Gouin-BonenfantTodaParetoExtrapolation}, which extrapolates the wealth distribution by a Pareto distribution outside the grid.\footnote{Readers interested in the detailed implementation are referred to \cite*{Gouin-BonenfantTodaParetoExtrapolation}. We use a 100-point affine-exponential grid for the asset in the range $a\in [0,10^4]$.}

Figure~\ref{fig:wDist} shows the stationary wealth distribution of the normalized wealth $\tilde{a}_t=a_t\e^{-gt}$ in log-log scale. The vertical axis shows the tail probability $\Pr(\tilde{a}_t>a)$ for thresholds $a\in [0,10^4]$. The log-log plots of the wealth distribution show a straight line pattern for high asset level, implying a power law behavior (\ie, $\Pr(\tilde{a}_t>a)\sim a^{-\alpha}$ for large $a$, where $\alpha>1$ is the Pareto exponent), which is consistent with theory. Letting $M$ be the matrix of conditional moment generating functions of log wealth growth defined by
\begin{equation}
M_{z\hat{z}}(\alpha)=\E (R(z,\hat{z},\hat{\zeta})(1-\bar{c}(z)))^\alpha,\label{eq:Mzz}
\end{equation}
using the formula in \cite*{BeareToda-dPL}, the Pareto exponent $\alpha$ solves
\begin{equation}
r(P\odot M(\alpha))=1.\label{eq:BeareToda}
\end{equation}
Numerically solving this equation, the Pareto exponents are $\alpha(2)=3.745$ for $\gamma=2$ and $\alpha(4)=1.714$ for $\gamma=4$. Thus the wealth distribution is more unequal (the Pareto exponent is smaller) when risk aversion is higher. This is because with $\gamma=4$, we have $\bar{c}(z)=0$, so $M_{z\hat{z}}(\alpha)$ in \eqref{eq:Mzz} becomes larger, which makes the solution to \eqref{eq:BeareToda} smaller. In the data, the U.S.\ wealth Pareto exponent is 1.52 \citep*[Table 8]{Vermeulen2018}, which is close to the value $\alpha(4)=1.714$ but much smaller than $\alpha(2)=3.745$. Therefore a model with zero asymptotic MPCs is potentially useful for explaining the observed high wealth inequality and small Pareto exponent.

\begin{figure}[!htb]
\centering
\includegraphics[width=0.7\linewidth]{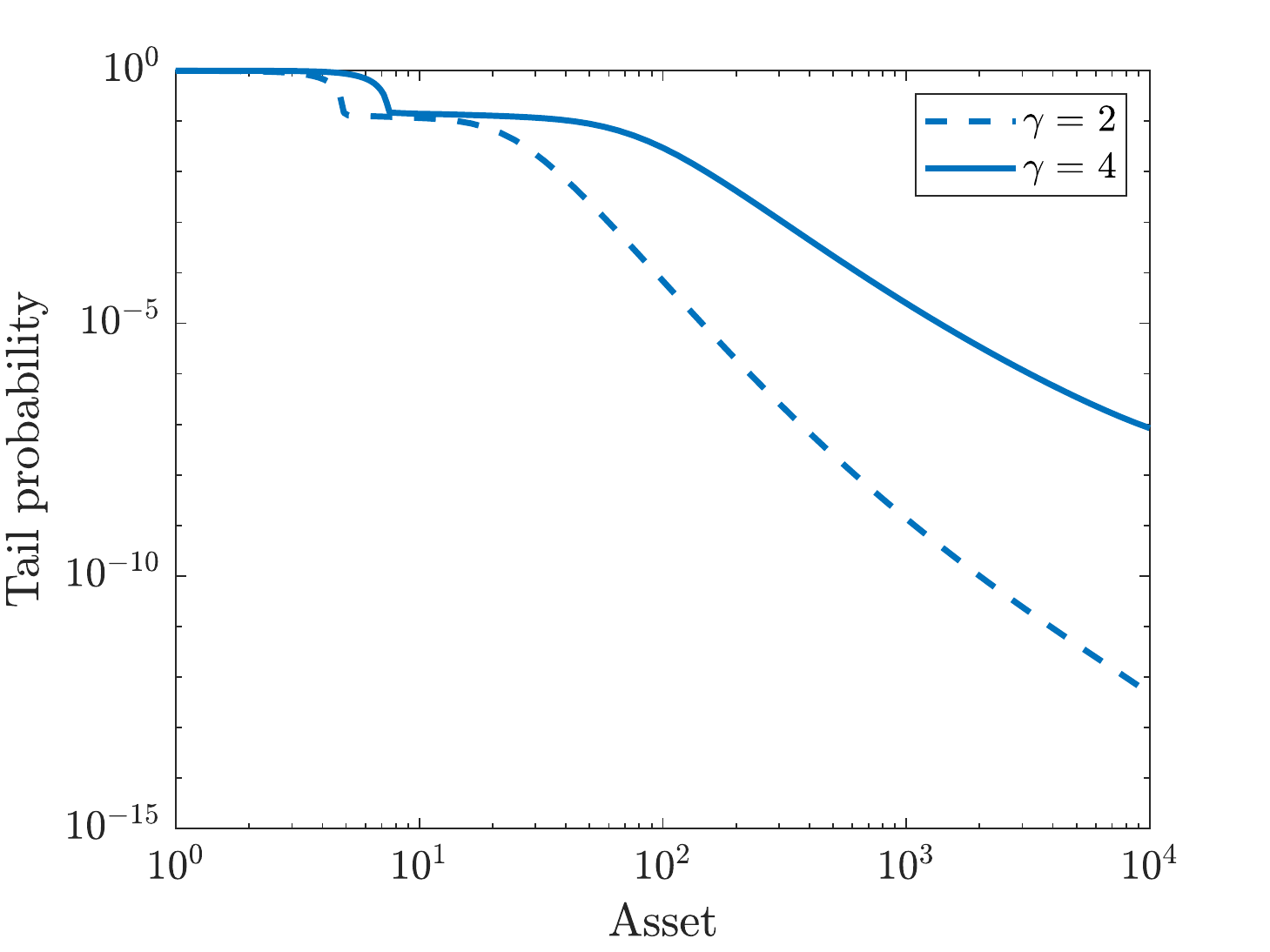}
\caption{Stationary distribution of normalized wealth $\tilde{a}_t=a_t\e^{-gt}$.}\label{fig:wDist}
\end{figure}

\section{Concluding remarks}

In this paper we have rigorously established that the homotheticity of preferences imply the asymptotic linearity of policy functions in a general income fluctuation problem. Furthermore, we have obtained an exact analytical characterization of the asymptotic marginal propensities to consume. Somewhat surprisingly, the asymptotic MPCs may become zero when capital loss is possible, implying a saving rate converging to 1 as agents get richer. Using a stylized model, we have demonstrated that zero asymptotic MPCs are empirically plausible. Our mechanism has the potential to accommodate a large saving rate of the rich and high wealth inequality (small Pareto exponent) as observed in the data without resorting to non-homothetic preferences or other frictions.

%\bibliographystyle{plainnat}
%\bibliography{localbib}

\begin{thebibliography}{47}
\providecommand{\natexlab}[1]{#1}
\providecommand{\url}[1]{\texttt{#1}}
\expandafter\ifx\csname urlstyle\endcsname\relax
  \providecommand{\doi}[1]{doi: #1}\else
  \providecommand{\doi}{doi: \begingroup \urlstyle{rm}\Url}\fi

\bibitem[A\c{c}{\i}kg{\"o}z(2018)]{Acikgoz2018}
{\"O}mer~T. A\c{c}{\i}kg{\"o}z.
\newblock On the existence and uniqueness of stationary equilibrium in {B}ewley
  economies with production.
\newblock \emph{Journal of Economic Theory}, 173:\penalty0 18--55, January
  2018.
\newblock \doi{10.1016/j.jet.2017.10.006}.

\bibitem[Alvarez and Stokey(1998)]{alvarez-stokey1998}
Fernando Alvarez and Nancy~L. Stokey.
\newblock Dynamic programming with homogeneous functions.
\newblock \emph{Journal of Economic Theory}, 82\penalty0 (1):\penalty0
  167--189, September 1998.
\newblock \doi{10.1006/jeth.1998.2431}.

\bibitem[Beare and Toda(2017)]{BeareToda-dPL}
Brendan~K. Beare and Alexis~Akira Toda.
\newblock Geometrically stopped {M}arkovian random growth processes and
  {P}areto tails.
\newblock 2017.
\newblock URL \url{https://arxiv.org/abs/1712.01431}.

\bibitem[Benhabib et~al.(2015)Benhabib, Bisin, and Zhu]{BenhabibBisinZhu2015}
Jess Benhabib, Alberto Bisin, and Shenghao Zhu.
\newblock The wealth distribution in {B}ewley economies with capital income
  risk.
\newblock \emph{Journal of Economic Theory}, 159\penalty0 (A):\penalty0
  489--515, September 2015.
\newblock \doi{10.1016/j.jet.2015.07.013}.

\bibitem[Bewley(1977)]{bewley1977}
Truman~F. Bewley.
\newblock The permanent income hypothesis: A theoretical formulation.
\newblock \emph{Journal of Economic Theory}, 16\penalty0 (2):\penalty0
  252--292, December 1977.
\newblock \doi{10.1016/0022-0531(77)90009-6}.

\bibitem[Borovi\v{c}ka and Stachurski(2020)]{BorovickaStachurski2020}
Jaroslav Borovi\v{c}ka and John Stachurski.
\newblock Necessary and sufficient conditions for existence and uniqueness of
  recursive utilities.
\newblock \emph{Journal of Finance}, 75\penalty0 (3):\penalty0 1457--1493, June
  2020.
\newblock \doi{10.1111/jofi.12877}.

\bibitem[Cagetti and De~Nardi(2006)]{CagettiDeNardi2006}
Marco Cagetti and Mariacristina De~Nardi.
\newblock Entrepreneurship, frictions, and wealth.
\newblock \emph{Journal of Political Economy}, 114\penalty0 (5):\penalty0
  835--870, October 2006.
\newblock \doi{10.1086/508032}.

\bibitem[Cao(2020)]{Cao2020}
Dan Cao.
\newblock Recursive equilibrium in {K}rusell and {S}mith (1998).
\newblock \emph{Journal of Economic Theory}, 186:\penalty0 104978, March 2020.
\newblock \doi{10.1016/j.jet.2019.104978}.

\bibitem[Carroll(2000)]{Carroll2000Why}
Christopher~D. Carroll.
\newblock Why do the rich save so much?
\newblock In Joel~B. Slemrod, editor, \emph{Does Atlas Shrug? The Economic
  Consequences of Taxing the Rich}, chapter~14, pages 465--484. Harvard
  University Press, Cambridge, MA, 2000.

\bibitem[Carroll(2006)]{Carroll2006}
Christopher~D. Carroll.
\newblock The method of endogenous gridpoints for solving dynamic stochastic
  optimization problems.
\newblock \emph{Economics Letters}, 91\penalty0 (3):\penalty0 312--320, June
  2006.
\newblock \doi{10.1016/j.econlet.2005.09.013}.

\bibitem[Carroll(2020)]{Carroll2020}
Christopher~D. Carroll.
\newblock Theoretical foundations of buffer stock saving.
\newblock \emph{Quantitative Economics}, 2020.
\newblock Forthcoming.

\bibitem[Carroll and Kimball(1996)]{CarrollKimball1996}
Christopher~D. Carroll and Miles~S. Kimball.
\newblock On the concavity of the consumption function.
\newblock \emph{Econometrica}, 64\penalty0 (4):\penalty0 981--992, July 1996.
\newblock \doi{10.2307/2171853}.

\bibitem[Chamberlain and Wilson(2000)]{ChamberlainWilson2000}
Gary Chamberlain and Charles~A. Wilson.
\newblock Optimal intertemporal consumption under uncertainty.
\newblock \emph{Review of Economic Dynamics}, 3\penalty0 (3):\penalty0
  365--395, July 2000.
\newblock \doi{10.1006/redy.2000.0098}.

\bibitem[Coleman(1990)]{Coleman1990}
Wilbur~John Coleman, II.
\newblock Solving the stochastic growth model by policy-function iteration.
\newblock \emph{Journal of Business and Economic Statistics}, 8\penalty0
  (1):\penalty0 27--29, January 1990.
\newblock \doi{10.1080/07350015.1990.10509769}.

\bibitem[Datta et~al.(2002)Datta, Mirman, and Reffett]{DattaMirmanReffett2002}
Manjira Datta, Leonard~J. Mirman, and Kevin~L. Reffett.
\newblock Existence and uniqueness of equilibrium in distorted dynamic
  economies with capital and labor.
\newblock \emph{Journal of Economic Theory}, 103\penalty0 (2):\penalty0
  377--410, April 2002.
\newblock \doi{10.1006/jeth.2000.2789}.

\bibitem[De~Nardi(2004)]{DeNardi2004}
Mariacristina De~Nardi.
\newblock Wealth inequality and intergenerational links.
\newblock \emph{Review of Economic Studies}, 71\penalty0 (3):\penalty0
  743--768, July 2004.
\newblock \doi{10.1111/j.1467-937X.2004.00302.x}.

\bibitem[Du(1990)]{Du1990}
Yihong Du.
\newblock Fixed points of increasing operators in ordered {B}anach spaces and
  applications.
\newblock \emph{Applicable Analysis}, 38\penalty0 (1-2):\penalty0 1--20, 1990.
\newblock \doi{10.1080/00036819008839957}.

\bibitem[Dynan et~al.(2004)Dynan, Skinner, and Zeldes]{DynanSkinnerZeldes2004}
Karen~E. Dynan, Jonathan Skinner, and Stephen~P. Zeldes.
\newblock Do the rich save more?
\newblock \emph{Journal of Political Economy}, 112\penalty0 (2):\penalty0
  397--444, April 2004.
\newblock \doi{10.1086/381475}.

\bibitem[Elsner et~al.(1988)Elsner, Johnson, and Dias~da
  Silva]{ElsnerJohnsonDiasDaSilva1988}
Ludwig Elsner, Charles~R. Johnson, and Jos{\'e}~Ant{\'o}nio Dias~da Silva.
\newblock The {P}erron root of a weighted geometric mean of nonneagative
  matrices.
\newblock \emph{Linear and Multilinear Algebra}, 24\penalty0 (1):\penalty0
  1--13, November 1988.
\newblock \doi{10.1080/03081088808817892}.

\bibitem[Fagereng et~al.(2019)Fagereng, Holm, Moll, and
  Natvik]{FagerengHolmMollNatvikWP}
Andreas Fagereng, Martin~Blomhoff Holm, Benjamin Moll, and Gisle Natvik.
\newblock Saving behavior across the wealth distribution: The importance of
  capital gains.
\newblock 2019.

\bibitem[Gilchrist et~al.(2009)Gilchrist, Yankov, and
  Zakraj{\v{s}}ek]{GilchristYankovZakrajsek2009}
Simon Gilchrist, Vladimir Yankov, and Egon Zakraj{\v{s}}ek.
\newblock Credit market shocks and economic fluctuations: Evidence from
  corporate bond and stock market.
\newblock \emph{Journal of Monetary Economics}, 56\penalty0 (4):\penalty0
  471--493, May 2009.
\newblock \doi{10.1016/j.jmoneco.2009.03.017}.

\bibitem[Gouin-Bonenfant and
  Toda(2018)]{Gouin-BonenfantTodaParetoExtrapolation}
{\'E}milien Gouin-Bonenfant and Alexis~Akira Toda.
\newblock {P}areto extrapolation: An analytical framework for studying tail
  inequality.
\newblock 2018.
\newblock URL \url{https://ssrn.com/abstract=3260899}.

\bibitem[Horn and Johnson(2013)]{HornJohnson2013}
Roger~A. Horn and Charles~R. Johnson.
\newblock \emph{Matrix Analysis}.
\newblock Cambridge University Press, New York, second edition, 2013.

\bibitem[Huggett(1996)]{huggett1996}
Mark Huggett.
\newblock Wealth distribution in life-cycle economies.
\newblock \emph{Journal of Monetary Economics}, 38\penalty0 (3):\penalty0
  469--494, December 1996.
\newblock \doi{10.1016/S0304-3932(96)01291-3}.

\bibitem[Kaplan et~al.(2018)Kaplan, Moll, and Violante]{KaplanMollViolante2018}
Greg Kaplan, Benjamin Moll, and Giovanni~L. Violante.
\newblock Monetary policy according to {HANK}.
\newblock \emph{American Economic Review}, 108\penalty0 (3):\penalty0 697--743,
  March 2018.
\newblock \doi{10.1257/aer.20160042}.

\bibitem[Krebs(2006)]{krebs2006}
Tom Krebs.
\newblock Recursive equilibrium in endogenous growth models with incomplete
  markets.
\newblock \emph{Economic Theory}, 29\penalty0 (3):\penalty0 505--523, 2006.
\newblock \doi{10.1016/S0165-1889(03)00062-9}.

\bibitem[Kuhn(2013)]{Kuhn2013}
Moritz Kuhn.
\newblock Recursive equilibria in an {A}iyagari-style economy with permanent
  income shocks.
\newblock \emph{International Economic Review}, 54\penalty0 (3):\penalty0
  807--835, August 2013.
\newblock \doi{10.1111/iere.12018}.

\bibitem[Lehrer and Light(2018)]{LehrerLight2018}
Ehud Lehrer and Bar Light.
\newblock The effect of interest rates on consumption in an income fluctuation
  problem.
\newblock \emph{Journal of Economic Dynamics and Control}, 94:\penalty0 63--71,
  September 2018.
\newblock \doi{10.1016/j.jedc.2018.07.004}.

\bibitem[Li and Stachurski(2014)]{LiStachurski2014}
Huiyu Li and John Stachurski.
\newblock Solving the income fluctuation problem with unbounded rewards.
\newblock \emph{Journal of Economic Dynamics and Control}, 45:\penalty0
  353--365, August 2014.
\newblock \doi{10.1016/j.jedc.2014.06.003}.

\bibitem[Light(2018)]{Light2018}
Bar Light.
\newblock Precautionary saving in a {M}arkovian earnings environment.
\newblock \emph{Review of Economic Dynamics}, 29:\penalty0 138--147, July 2018.
\newblock \doi{10.1016/j.red.2017.12.004}.

\bibitem[Light(2020)]{Light2020}
Bar Light.
\newblock Uniqueness of equilibrium in a {B}ewley-{A}iyagari economy.
\newblock \emph{Economic Theory}, 69:\penalty0 435--450, 2020.
\newblock \doi{10.1007/s00199-018-1167-z}.

\bibitem[Ma et~al.(2020)Ma, Stachurski, and Toda]{MaStachurskiToda2020JET}
Qingyin Ma, John Stachurski, and Alexis~Akira Toda.
\newblock The income fluctuation problem and the evolution of wealth.
\newblock \emph{Journal of Economic Theory}, 187:\penalty0 105003, May 2020.
\newblock \doi{10.1016/j.jet.2020.105003}.

\bibitem[McDaniel(2007)]{mcdaniel2007average}
Cara McDaniel.
\newblock Average tax rates on consumption, investment, labor and capital in
  the {OECD} 1950-2003.
\newblock 2007.

\bibitem[Mian et~al.(2020)Mian, Straub, and Sufi]{MianStraubSufiWP}
Atif~R. Mian, Ludwig Straub, and Amir Sufi.
\newblock Indebted demand.
\newblock NBER Working Paper 26940, 2020.
\newblock URL \url{https://www.nber.org/papers/w26940}.

\bibitem[Morand and Reffett(2003)]{MorandReffett2003}
Olivier~F. Morand and Kevin~L. Reffett.
\newblock Existence and uniqueness of equilibrium in nonoptimal unbounded
  infinite horizon economies.
\newblock \emph{Journal of Monetary Economics}, 50\penalty0 (6):\penalty0
  1351--1373, September 2003.
\newblock \doi{10.1016/S0304-3932(03)00082-5}.

\bibitem[Quadrini(1999)]{Quadrini1999}
Vincenzo Quadrini.
\newblock The importance of entrepreneurship for wealth concentration and
  mobility.
\newblock \emph{Review of Income and Wealth}, 45\penalty0 (1):\penalty0 1--19,
  March 1999.
\newblock \doi{10.1111/j.1475-4991.1999.tb00309.x}.

\bibitem[Rabault(2002)]{Rabault2002}
Guillaume Rabault.
\newblock When do borrowing constraints bind? {S}ome new results on the income
  fluctuation problem.
\newblock \emph{Journal of Economic Dynamics and Control}, 26\penalty0
  (2):\penalty0 217--245, February 2002.
\newblock \doi{10.1016/S0165-1889(00)00042-7}.

\bibitem[Saez and Zucman(2016)]{saez-zucman2016}
Emmanuel Saez and Gabriel Zucman.
\newblock Wealth inequality in the {U}nited {S}tates since 1913: Evidence from
  capitalized income tax data.
\newblock \emph{Quarterly Journal of Economics}, 131\penalty0 (2):\penalty0
  519--578, May 2016.
\newblock \doi{10.1093/qje/qjw004}.

\bibitem[Samuelson(1969)]{samuelson1969}
Paul~A. Samuelson.
\newblock Lifetime portfolio selection by dynamic stochastic programming.
\newblock \emph{Review of Economics and Statistics}, 51\penalty0 (3):\penalty0
  239--246, August 1969.
\newblock \doi{10.2307/1926559}.

\bibitem[Stachurski and Toda(2019)]{StachurskiToda2019JET}
John Stachurski and Alexis~Akira Toda.
\newblock An impossibility theorem for wealth in heterogeneous-agent models
  with limited heterogeneity.
\newblock \emph{Journal of Economic Theory}, 182:\penalty0 1--24, July 2019.
\newblock \doi{10.1016/j.jet.2019.04.001}.

\bibitem[Straub(2019)]{StraubSavingWP}
Ludwig Straub.
\newblock Consumption, savings, and the distribution of permanent income.
\newblock 2019.

\bibitem[Toda(2014)]{Toda2014JET}
Alexis~Akira Toda.
\newblock Incomplete market dynamics and cross-sectional distributions.
\newblock \emph{Journal of Economic Theory}, 154:\penalty0 310--348, November
  2014.
\newblock \doi{10.1016/j.jet.2014.09.015}.

\bibitem[Toda(2019)]{Toda2019JME}
Alexis~Akira Toda.
\newblock Wealth distribution with random discount factors.
\newblock \emph{Journal of Monetary Economics}, 104:\penalty0 101--113, June
  2019.
\newblock \doi{10.1016/j.jmoneco.2018.09.006}.

\bibitem[Toda(2020)]{TodaHARA}
Alexis~Akira Toda.
\newblock Necessity of hyperbolic absolute risk aversion for the concavity of
  consumption functions.
\newblock \emph{Journal of Mathematical Economics}, 2020.
\newblock \doi{10.1016/j.jmateco.2020.102460}.

\bibitem[Vermeulen(2018)]{Vermeulen2018}
Philip Vermeulen.
\newblock How fat is the top tail of the wealth distribution?
\newblock \emph{Review of Income and Wealth}, 64\penalty0 (2):\penalty0
  357--387, June 2018.
\newblock \doi{10.1111/roiw.12279}.

\bibitem[Welch and Goyal(2008)]{welch-goyal2008}
Ivo Welch and Amit Goyal.
\newblock A comprehensive look at the empirical performance of equity premium
  prediction.
\newblock \emph{Review of Financial Studies}, 21\penalty0 (4):\penalty0
  1455--1508, July 2008.
\newblock \doi{10.1093/rfs/hhm014}.

\bibitem[Zhang(2013)]{Zhang2013}
Zhitao Zhang.
\newblock \emph{Variational, Topological, and Partial Order Methods with Their
  Applications}, volume~29 of \emph{Developments in Mathematics}.
\newblock Springer, 2013.
\newblock \doi{10.1007/978-3-642-30709-6}.

\end{thebibliography}

\appendix

\section{Solving the income fluctuation problem}\label{sec:IF}

\begin{proof}[Proof of Lemma~\ref{lem:xi}]
Since $c\in \cC$, by \eqref{eq:uprimediff} we have
\begin{equation*}
M\coloneqq \sup_{(a,z)\in (0,\infty)\times \ZZ}\abs{u'(c(a,z))-u'(a)}<\infty.
\end{equation*}
Since $\beta,R,Y\ge 0$, $c\in \cC$ is increasing in its first argument, and $u'$ is decreasing, the function $\xi\mapsto \hat{\beta}\hat{R}u'(c(\hat{R}(a-\xi)+\hat{Y},\hat{Z}))$ is increasing. Hence for $\xi\in [0,a]$, we have
\begin{equation}
0\le \hat{\beta}\hat{R}u'(c(\hat{R}(a-\xi)+\hat{Y},\hat{Z}))\le \hat{\beta}\hat{R}u'(c(\hat{Y},\hat{Z}))\le \hat{\beta}\hat{R}[u'(\hat{Y})+M].\label{eq:betaRuY}
\end{equation}
Using the (constant) function $\hat{\beta}\hat{R}[u'(\hat{Y})+M]$ as the dominating function, an application of the dominated convergence theorem together with Assumption~\ref{asmp:spectral}\ref{item:Ebeta}\ref{item:EY} implies that
\begin{equation*}
\xi\mapsto \E_{z,\hat{z}}\hat{\beta}\hat{R}u'(c(\hat{R}(a-\xi)+\hat{Y},\hat{Z}))
\end{equation*}
is finite, continuous, and increasing in $\xi\in [0,a]$. Therefore
\begin{equation*}
\E_z\hat{\beta}\hat{R}u'(c(\hat{R}(a-\xi)+\hat{Y},\hat{Z}))=\sum_{\hat{z}=1}^ZP_{z\hat{z}}\E_{z,\hat{z}}\hat{\beta}\hat{R}u'(c(\hat{R}(a-\xi)+\hat{Y},\hat{Z}))
\end{equation*}
is also finite, continuous, and increasing in $\xi\in [0,a]$. Noting that $u'$ is continuous and strictly decreasing on $(0,\infty)$,
\begin{equation*}
g(\xi)\coloneqq u'(\xi)-\min\set{\max\set{\E_z \hat{\beta}\hat{R}u'(c(\hat{R}(a-\xi)+\hat{Y},\hat{Z})),u'(a)},u'(0)}
\end{equation*}
is continuous and strictly decreasing on $(0,a]$, and it is also continuous at $\xi=0$ if $u'(0)<\infty$. Since $u'(a)\le u'(0)$, we have
\begin{equation*}
g(a)\le u'(a)-\min\set{u'(a),u'(0)}=u'(a)-u'(a)=0.
\end{equation*}
If $u'(0)<\infty$, then $g(0)\ge u'(0)-u'(0)=0$. If $u'(0)=\infty$, then using \eqref{eq:betaRuY},
\begin{align*}
g(\xi)&=u'(\xi)-\max\set{\E_z \hat{\beta}\hat{R}u'(c(\hat{R}(a-\xi)+\hat{Y},\hat{Z})),u'(a)}\\
&\ge u'(\xi)-\max\set{\E_z \hat{\beta}\hat{R}[u'(\hat{Y})+M],u'(a)}\to \infty
\end{align*}
as $\xi\downarrow 0$. Therefore by the intermediate value theorem, there exists $\xi\in [0,a]$ with $g(\xi)=0$ (with $\xi>0$ if $u'(0)=\infty$), and $\xi$ is unique because $g$ is strictly decreasing.
\end{proof}

The proof of Theorem~\ref{thm:exist} is long and technical, but very similar to the proof of Theorem 2.2 of \cite*{MaStachurskiToda2020JET}. Therefore we only provide a sketch of the proof and explain how our weaker assumptions can be handled in a similar way.

To construct a contraction mapping, it is convenient to work in the space of functions $h:(0,\infty)\to \R^Z$ defined by $h(a)=(h_1(a),\dots,h_Z(a))$ with $h_z(a)\coloneqq u'(c(a,z))$. Noting that $u'$ is continuous and strictly decreasing, we can easily see from the definition of $\cC$ that each $h_z$ is continuous, decreasing, and $h_z(a)-u'(a)$ is nonnegative and bounded (see \eqref{eq:uprimediff}). Therefore define the space $\cH$ by the set of functions $h:(0,\infty)\to \R^Z$ such that each $h_z$ is continuous, decreasing, $h_z(a)-u'(a)\ge 0$ for all $a>0$, and
\begin{equation*}
\sup_{a\in (0,\infty)}\abs{h_z(a)-u'(a)}<\infty.
\end{equation*}
For $h^1,h^2\in \cH$, if we define
\begin{equation*}
\tilde{\rho}(h^1,h^2)=\max_{z\in \ZZ}\sup_{a\in (0,\infty)}\abs{h^1_z(a)-h^2_z(a)},
\end{equation*}
then $(\cH,\tilde{\rho})$ becomes a complete metric space. For $h\in \cH$, $a>0$, and $z\in \ZZ$, define the function $\tilde{T}h:(0,\infty)\to \R^Z$ by
\begin{equation*}
(\tilde{T}h)_z(a)=u'(Tc(a,z)),
\end{equation*}
where $Tc(a,z)$ is the unique $\xi\in [0,a]$ solving \eqref{eq:xi}, whose existence and uniqueness is established in Lemma~\ref{lem:xi}. Then letting $\kappa=(\tilde{T}h)_z(a)=u'(\xi)$, it follows from \eqref{eq:xi} that
\begin{equation}
\kappa=\min\set{\max\set{\E_z \hat{\beta}\hat{R}h_{\hat{Z}}(\hat{R}(a-(u')^{-1}(\kappa))+\hat{Y}),u'(a)},u'(0)}.\label{eq:kappa}
\end{equation}
Using a similar argument to the proofs of Proposition B.4 and Lemma B.3 of \cite*{MaStachurskiToda2020JET}, we can show that $\tilde{T}$ is a monotone self map on $\cH$, \ie, $\tilde{T}:\cH\to \cH$ and $h^1\le h^2$ implies $\tilde{T}h^1\le \tilde{T}h^2$. The following lemma is useful for establishing that $\tilde{T}$ has a contraction property. Below, for $h\in \cH$ and $v\in \R_+^Z$, define $h+v\in \cH$ by $(h+v)_z(a)=h_z(a)+v_z$.

\begin{lem}\label{lem:hv}
Let $K$ be as in \eqref{eq:Ktheta}. For any $h\in \cH$ and $v\in \R_+^Z$, we have
\begin{equation}
\tilde{T}(h+v)\le \tilde{T}h+K(1)v.\label{eq:Tdiscount}
\end{equation}
\end{lem}

\begin{proof}
If $x,y,z\in\R$ and $\alpha\ge 0$, note that
\begin{align}
\min\set{\max\set{x+\alpha,y},z}&\le \min\set{\max\set{x+\alpha,y+\alpha},z+\alpha}\notag \\
&=\min\set{\max\set{x,y},z}+\alpha.\label{eq:minmaxineq}
\end{align}
Letting $\kappa_v\coloneqq (\tilde{T}(h+v))_z(a)$ in \eqref{eq:kappa}, using \eqref{eq:minmaxineq}, and recalling the definition of $K$ in \eqref{eq:Ktheta}, we obtain
\begin{align*}
&(\tilde{T}(h+v))_z(a)=\kappa_v\\
&=\min\set{\max\set{\E_z \hat{\beta}\hat{R}(h_{\hat{Z}}+v_{\hat{Z}})(\hat{R}(a-(u')^{-1}(\kappa_v))+\hat{Y}),u'(a)},u'(0)}\\
&\le \min\set{\max\set{\E_z \hat{\beta}\hat{R}h_{\hat{Z}}(\hat{R}(a-(u')^{-1}(\kappa_v))+\hat{Y}),u'(a)},u'(0)}+(K(1)v)_z.
\end{align*}
Therefore to show \eqref{eq:Tdiscount}, it suffices to show
\begin{equation}
\min\set{\max\set{\E_z \hat{\beta}\hat{R}h_{\hat{Z}}(\hat{R}(a-(u')^{-1}(\kappa_v))+\hat{Y}),u'(a)},u'(0)}\le (\tilde{T}h)_z(a).\label{eq:suffineq1}
\end{equation}
Noting that $\kappa\coloneqq (\tilde{T}h)_z(a)$ satisfies \eqref{eq:kappa}, to show \eqref{eq:suffineq1}, it suffices to show
\begin{equation}
\E_z \hat{\beta}\hat{R}h_{\hat{Z}}(\hat{R}(a-(u')^{-1}(\kappa_v))+\hat{Y})\le \E_z \hat{\beta}\hat{R}h_{\hat{Z}}(\hat{R}(a-(u')^{-1}(\kappa))+\hat{Y}).\label{eq:suffineq2}
\end{equation}
Since $\tilde{T}$ is monotone and $h\le h+v$, we have $\kappa=(\tilde{T}h)_z(a)\le (\tilde{T}(h+v))_a(z)=\kappa_v$. Since $u'$ (hence $(u')^{-1}$) is strictly decreasing, we obtain
\begin{equation*}
a-(u')^{-1}(\kappa)\le a-(u')^{-1}(\kappa_v).
\end{equation*}
Since $\hat{\beta},\hat{R},\hat{Y}\ge 0$, \eqref{eq:suffineq2} holds because $h$ is decreasing.
\end{proof}

Using Lemma~\ref{lem:hv}, we can show that $\tilde{T}^k$ is a contraction for some $k\in \N$.

\begin{lem}\label{lem:Tcontract}
If Assumptions~\ref{asmp:Inada} and \ref{asmp:spectral} hold, then there exists $k\in \N$ such that $\tilde{T}^k$ is a contraction on $\cH$. Consequently, $\tilde{T}$ has a unique fixed point $h^*\in \cH$ and $\tilde{T}^nh^0\to h^*$ as $n\to\infty$ for any $h^0\in \cH$.
\end{lem}

\begin{proof}
Take any $h^1,h^2\in \cH$. Define $v\in \R_+^Z$ by 
\begin{equation*}
v_z=\sup_{a\in (0,\infty)}\abs{h^1_z(a)-h^2_z(a)}<\infty.
\end{equation*}
Then clearly $h^1\le h^2+v$, so a repeated application of Lemma~\ref{lem:hv} and the monotonicity of $\tilde{T}$ imply $\tilde{T}^kh^1\le \tilde{T}^kh^2+K(1)^kv$ for all $k$. Interchanging $h^1,h^2$, it follows that
\begin{equation*}
\abs{(\tilde{T}^kh^1)_z(a)-(\tilde{T}^kh^2)_z(a)}\le (K(1)^kv)_z
\end{equation*}
for all $k$, $a>0$, and $z\in \ZZ$. Taking the supremum over $a\in (0,\infty)$ and $z\in \ZZ$ and letting $\norm{\cdot}$ be the supremum norm on $\R^Z$ (and the induced matrix norm for $Z\times Z$ matrices), it follows that
\begin{equation*}
\tilde{\rho}(\tilde{T}^kh^1,\tilde{T}^kh^2)\le \norm{K(1)^k}\norm{v}=\norm{K(1)^k}\tilde{\rho}(h^1,h^2)
\end{equation*}
for all $k$. By the Gelfand spectral radius formula \cite[Theorem 5.7.10]{HornJohnson2013}, we have $\norm{K(1)^k}^{1/k}\to r(K(1))<1$ as $k\to\infty$ by Assumption~\ref{asmp:spectral}\ref{item:rbeta}. In particular, there exists $k\in \N$ such that $\norm{K(1)^k}<1$, which implies that $\tilde{T}^k$ is a contraction.
\end{proof}

The rest of the proof of Theorem~\ref{thm:exist} is similar to \cite*{MaStachurskiToda2020JET}. Letting $h^*\in \cH$ be the unique fixed point of $\tilde{T}$ and defining $c(a,z)=(u')^{-1}(h_z^*(a))$, we can easily verify that $\xi=c(a,z)$ satisfies the Euler equation \eqref{eq:xi}. Furthermore, $\tilde{T}^nh^0\to h^*$ for all $h^0\in \cH$ implies $T^nc_0\to c$ for all $c_0\in \cC$. Using the analogues of Lemma B.1, Lemma B.2, Proposition B.1, and Proposition 2.2 of \cite*{MaStachurskiToda2020JET}, it follows that $c(a,z)$ is the unique optimal consumption function. (The remaining conditions $r(K(0))<1$, $u'(\infty)<1$, and $\E_{z,\hat{z}}\hat{Y}<\infty$ are used to show that the value function is finite and the transversality condition holds.)

\section{Proof of main results}\label{sec:proof}

The proof of Theorem~\ref{thm:linear} is technical and consists of the following steps:
\begin{enumerate}
\item show that policy function iteration leads to increasingly tighter upper bounds on consumption functions that are asymptotically linear with explicit slopes,
\item show that the slopes of the upper bounds converge using the fixed point theory of monotone convex maps, and
\item show that the consumption functions have linear lower bounds with identical slopes to the limit of upper bounds, implying asymptotic linearity.
\end{enumerate}

Let $\cC$ be the space of candidate consumption functions and $T:\cC\to \cC$ be the time iteration operator as defined in Section~\ref{sec:AL}. Since the CRRA utility satisfies $u'(c)=c^{-\gamma}$ and hence $u'(0)=\infty$, by Lemma~\ref{lem:xi} $\xi=Tc(a,z)$ satisfies $\xi\in (0,a]$. The following proposition allows us to asymptotically bound the consumption rate $c(a,z)/a$ from above.

\begin{prop}\label{prop:UB}
Let everything be as in Theorem~\ref{thm:linear}. If $c\in \cC$ and
$$\limsup_{a\to\infty}\frac{c(a,z)}{a}\le x(z)^{-1/\gamma}$$
for some $x(z)\ge 1$ for all $z\in \ZZ$, then
\begin{equation}
\limsup_{a\to\infty}\frac{Tc(a,z)}{a}\le (Fx)(z)^{-1/\gamma}.\label{eq:Tcbar}
\end{equation}
\end{prop}

\begin{proof}
Let $\alpha=\limsup_{a\to\infty}Tc(a,z)/a$. By definition, we can take an increasing sequence $\set{a_n}$ such that $\alpha=\lim_{n\to\infty} Tc(a_n,z)/a_n$. Define $\alpha_n=Tc(a_n,z)/a_n\in (0,1]$ and
\begin{equation}
\lambda_n=\frac{c(\hat{R}(1-\alpha_n)a_n+\hat{Y},\hat{Z})}{a_n}>0.\label{eq:lambdan}
\end{equation}
Let us show that
\begin{equation}
\limsup_{n\to\infty}\lambda_n\le x(\hat{Z})^{-1/\gamma}\hat{R}(1-\alpha).\label{eq:lambdalim}
\end{equation}
To see this, if $\alpha<1$ and $\hat{R}>0$, then since $\hat{R}(1-\alpha_n)a_n\to \hat{R}(1-\alpha)\cdot \infty=\infty$, by assumption we have
\begin{align*}
\limsup_{n\to\infty}\lambda_n&=\limsup_{n\to\infty}\frac{c(\hat{R}(1-\alpha_n)a_n+\hat{Y},\hat{Z})}{\hat{R}(1-\alpha_n)a_n+\hat{Y}}\left(\hat{R}(1-\alpha_n)+\frac{\hat{Y}}{a_n}\right)\\
&\le \limsup_{a\to\infty}\frac{c(a,\hat{Z})}{a}\times \hat{R}(1-\alpha)\\
&\le x(\hat{Z})^{-1/\gamma}\hat{R}(1-\alpha),
\end{align*}
which is \eqref{eq:lambdalim}. If $\alpha=1$ or $\hat{R}=0$, then since $c(a,z)\le a$, we have
\begin{align*}
\lambda_n&=\frac{c(\hat{R}(1-\alpha_n)a_n+\hat{Y},\hat{Z})}{\hat{R}(1-\alpha_n)a_n+\hat{Y}}\left(\hat{R}(1-\alpha_n)+\frac{\hat{Y}}{a_n}\right)\\
&\le \hat{R}(1-\alpha_n)+\frac{\hat{Y}}{a_n}\to \hat{R}(1-\alpha)=0,
\end{align*}
so again \eqref{eq:lambdalim} holds.

Since $\xi_n\coloneqq Tc(a_n,z)=\alpha_na_n$ solves the Euler equation \eqref{eq:xi}, using $u'(c)=c^{-\gamma}$, $u'(0)=\infty$, and the definition of $\lambda_n$ in \eqref{eq:lambdan}, we have
\begin{align}
0&=\frac{u'(\alpha_n a_n)}{u'(a_n)}-\max\set{\E_z \hat{\beta}\hat{R}\frac{u'(c(\hat{R}(1-\alpha_n)a_n+\hat{Y},\hat{Z}))}{u'(a_n)},1}\notag \\
&=\alpha_n^{-\gamma}-\max\set{\E_z \hat{\beta}\hat{R}(c(\hat{R}(1-\alpha_n)a_n+\hat{Y},\hat{Z})/a_n)^{-\gamma},1} \notag \\
&=\alpha_n^{-\gamma}-\max\set{\E_z\hat{\beta}\hat{R}\lambda_n^{-\gamma},1}\notag \\
\implies \alpha_n^{-\gamma}&=\max\set{\E_z\hat{\beta}\hat{R}\lambda_n^{-\gamma},1}\ge \E_z\hat{\beta}\hat{R}\lambda_n^{-\gamma}.\label{eq:eulern}
\end{align}

Now letting $n\to\infty$ in \eqref{eq:eulern} and applying Fatou's lemma, we obtain
\begin{align*}
\alpha^{-\gamma}=\lim_{n\to\infty}\alpha_n^{-\gamma}&\ge \liminf_{n\to\infty} \E_z\hat{\beta}\hat{R}\lambda_n^{-\gamma}\\
&\ge \E_z \liminf_{n\to\infty}\hat{\beta}\hat{R}\lambda_n^{-\gamma}\\
&=\E_z\hat{\beta}\hat{R}\left[\limsup_{n\to\infty} \lambda_n\right]^{-\gamma}\\
&\ge \E_z\hat{\beta}\hat{R}\left[x(\hat{Z})^{-1/\gamma}\hat{R}(1-\alpha)\right]^{-\gamma}
\end{align*}
by \eqref{eq:lambdalim}. Solving the inequality for $\alpha$ and using the convention $\beta R^{1-\gamma}=(\beta R)R^{-\gamma}$ and $0\cdot \infty=0$ (see Footnote~\ref{fn:convention}), we obtain
$$\limsup_{a\to\infty}\frac{Tc(a,z)}{a}=\alpha\le \frac{1}{1+\left(\E_z \hat{\beta}\hat{R}^{1-\gamma}x(\hat{Z})\right)^{1/\gamma}}=(Fx)(z)^{-1/\gamma}. \qedhere$$
\end{proof}

Starting from the trivial upper bound $c(a,z)\le a$ and applying Proposition~\ref{prop:UB} repeatedly, we obtain increasingly tighter upper bounds of $c(a,z)$. The following proposition characterizes the limits of the slopes of the upper bounds.

\begin{prop}\label{prop:MPC_seq}
Let everything be as in Theorem~\ref{thm:linear}. Then $F$ in \eqref{eq:fixedpoint} has a fixed point $x^*\in \R_+^Z$ if and only if $r(K(1-\gamma))<1$, in which case the fixed point is unique. Take any $x_0\in \R_+^Z$ and define the sequence $\set{x_n}_{n=1}^\infty\subset \R_+^Z$ by
\begin{equation}
x_n=Fx_{n-1} \label{eq:xn}
\end{equation}
for all $n\in \N$. Then the following statements are true:
\begin{enumerate}
\item If $r(K(1-\gamma))<1$, then $\set{x_n}_{n=1}^\infty$ converges to $x^*$.
\item If $r(K(1-\gamma))\ge 1$ and $K(1-\gamma)$ is irreducible, then $x_n(z)\to x^*(z)=\infty$ as $n\to\infty$ for all $z\in \ZZ$.
\end{enumerate}
\end{prop}

\begin{proof}
Immediate from Lemmas~\ref{lem:phi} and \ref{lem:xn} below.
\end{proof}

\begin{lem}\label{lem:phi}
Let $\gamma>0$ and define $\phi:\R_+\to\R_+$ by $\phi(t)=(1+t^{1/\gamma})^\gamma$. Then there exist $a\ge 1$ and $b\ge 0$ such that $\phi(t)\le at+b$. Furthermore, we can take $a\ge 1$ arbitrarily close to 1. (The choice of $b$ may depend on $a$.)
\end{lem}

\begin{proof}
The proof depends on $\gamma \gtrless 1$.
\begin{case}[$\gamma\le 1$]
Let us show that we can take $a=b=1$. Let $f(t)=1+t-\phi(t)$. Then $f(0)=0$ and
$$f'(t)=1-\phi'(t)=1-\gamma(1+t^{1/\gamma})^{\gamma-1}\frac{1}{\gamma}t^{1/\gamma-1}=1-(t^{-1/\gamma}+1)^{\gamma-1}\ge 0,$$
so $f(t)\ge 0$ for all $t\ge 0$. Therefore $\phi(t)\le 1+t$.
\end{case}
\begin{case}[$\gamma>1$]
By simple algebra we obtain
\begin{equation}
\phi''(t)=(\gamma-1)(t^{-1/\gamma}+1)^{\gamma-2}\left(-\frac{1}{\gamma}t^{-1/\gamma-1}\right)<0,\label{eq:phi''}
\end{equation}
so $\phi$ is increasing and concave. Therefore $\phi(t)\le \phi(u)+\phi'(u)(t-u)$ for all $t,u$. Letting $a=\phi'(u)$ and $b=\max\set{0,\phi(u)-\phi'(u)u}$, we obtain $\phi(t)\le at+b$. Furthermore, since $\phi'(t)=(t^{-1/\gamma}+1)^{\gamma-1}\to 1$ as $t\to\infty$, we can take $a=\phi'(u)$ arbitrarily close to 1 by taking $u$ large enough.\qedhere
\end{case}
\end{proof}

%The following lemma slightly generalizes a result in \cite*{Toda2019JME}.

\begin{lem}\label{lem:xn}
Let $\gamma>0$ and $K$ be a $Z \times Z$ nonnegative matrix. Define $F:\R_+^Z\to \R_+^Z$ by $Fx=\phi(Kx)$, where $\phi$ is as in Lemma~\ref{lem:phi} and is applied entry-wise. Then $F$ has a fixed point $x^*\in \R_+^Z$ if and only if $r(K)<1$, in which case $x^*$ is unique.

Take any $x_0\in \R_+^Z$ and define the sequence $\set{x_n}_{n=1}^\infty\subset \R_+^Z$ by $x_n=Fx_{n-1}$ for all $n\in \N$. Then the following statements are true:
\begin{enumerate}
\item If $r(K)<1$, then $\set{x_n}_{n=1}^\infty$ converges to $x^*$.
\item If $r(K)\ge 1$ and $K$ is irreducible, then $x_n(z)\to x^*(z)=\infty$ as $n\to\infty$ for all $z\in \ZZ$.
\end{enumerate}
\end{lem}

\begin{proof}
We divide the proof into three steps.

\begin{step}
If $r(K)\ge 1$, then $F$ does not have a fixed point. If in addition $K$ is irreducible, then $x_n(z)\to\infty$ for all $z\in \Z$.
\end{step}

We prove the contrapositive. Suppose that $F$ has a fixed point $x^*\in \R_+^Z$. Since $\phi>0$, we have $x^*\gg 0$. Since clearly $\phi(t)>t$ for all $t\ge 0$, we have $x^*=\phi(Kx^*)\gg Kx^*$. Since $K$ is a nonnegative matrix, by the Perron-Frobenius theorem, we can take a right eigenvector $y>0$ such that $y' K=r(K)y'$. Since $x^*\gg Kx^*$ and $y>0$, we obtain $r(K)y'x^*=y'Kx^*<y'x^*$.
Dividing both sides by $y' x^*>0$, we obtain $r(K)<1$.

Suppose that $r(K)\ge 1$ and $K$ is irreducible. Since $K$ is nonnegative and $\phi$ is strictly increasing, $F=\phi\circ K$ is a monotone map. Therefore to show $x_n(z)\to\infty$, it suffices to show this when $x_0=0$. Since $x_1=Fx_0=F0=1\ge 0$, applying $F^{n-1}$ we obtain $x_n\ge x_{n-1}$ for all $n$. Since $\set{x_n}_{n=0}^\infty$ is an increasing sequence in $\R_+^Z$, if it is bounded, then it converges to some $x^*\in \R_+^Z$. By continuity, $x^*$ is a fixed point of $F$, which is a contradiction. Therefore $\set{x_n}_{n=0}^\infty$ is unbounded, so $x_n(\hat{z})\to\infty$ for at least one $\hat{z}\in \ZZ$. Since by assumption $K$ is irreducible, for each $(z,\hat{z})\in \ZZ^2$, there exists $m\in \N$ such that $K^m_{z\hat{z}}>0$. Therefore
$$x_{m+n}(z)\ge K^m_{z\hat{z}}x_n(\hat{z})\to \infty$$
as $n\to\infty$, so $x_n(z)\to\infty$ for all $z\in \ZZ$.

\begin{step}
If $r(K)<1$, then $F$ has a unique fixed point $x^*$ in $\R_+^Z$. If we take $a\in [1,1/r(K))$ and $b>0$ as in Lemma~\ref{lem:phi}, then
\begin{equation}
1\le x^*\ll (I-aK)^{-1}b1.\label{eq:x*ub}
\end{equation}
\end{step}

Take any fixed point $x^*\in \R_+^Z$ of $F$. Since $\phi(t)\ge 1$ for all $t\ge 0$, clearly $x^*\ge 1$. Since $K$ is nonnegative and $ar(K)<1$, the inverse $(I-aK)^{-1}=\sum_{k=0}^\infty (aK)^k$ exists and is nonnegative. Therefore
$$x^*=Fx^*\ll aKx^*+b1\implies x^*\ll (I-aK)^{-1}b1,$$
which is \eqref{eq:x*ub}.

The proof of existence and uniqueness uses a similar strategy to \cite*{BorovickaStachurski2020}. Clearly $F$ is a monotone map. Using \eqref{eq:phi''}, it follows that $F$ is convex if $\gamma\le 1$ and concave if $\gamma\ge 1$. Define $u_0=0$ and $v_0=(I-aK)^{-1}b1\gg 0$. Then $Fu_0=1\gg 0=u_0$ and $Fv_0=\phi(Kv_0)\ll aKv_0+b1=v_0$. Hence by Theorem 2.1.2 of \cite*{Zhang2013}, which is based on Theorem 3.1 of \cite*{Du1990}, $F$ has a unique fixed point in $[u_0,v_0]=[0,v_0]$. Since by \eqref{eq:x*ub} any fixed point $x^*$ must lie in this interval, it follows that $F$ has a unique fixed point in $\R_+^Z$.

\begin{step}
If $r(K)<1$, then $\set{x_n}_{n=1}^\infty$ converges to $x^*$.
\end{step}

Let $a\in [1,1/r(K))$, $b>0$, and $v_0\gg 0$ be as in the previous step. Since $Fx=\phi(Kx)$, we obtain
$$x_n=Fx_{n-1}=\phi(Kx_{n-1})\ll aKx_{n-1}+b1.$$
Iterating, we obtain
\begin{align*}
x_n&\ll (aK)^nx_0+\sum_{k=0}^{n-1}(aK)^k(b1)\\
&=(aK)^nx_0+\sum_{k=0}^\infty (aK)^k(b1)-\sum_{k=n}^\infty(aK)^k(b1)\\
&=(aK)^n(x_0-v_0)+v_0.
\end{align*}
Since $r(aK)=ar(K)<1$, we have $(aK)^n(x_0-v_0)\to 0$ as $n\to\infty$. Therefore $0=u_0\ll x_n\ll v_0$ for large enough $n$. Again by Theorem 2.1.2 of \cite*{Zhang2013}, we have $x_n\to x^*$ as $n\to\infty$.
\end{proof}

The following proposition allows us to obtain explicit linear lower bounds on consumption functions.

\begin{prop}\label{prop:LB}
Let everything be as in Theorem~\ref{thm:linear}. Suppose $r(K(1-\gamma))<1$ and let $x^*\in \R_{++}^Z$ be the unique fixed point of $F$ in \eqref{eq:fixedpoint}. Restrict the candidate space to
\begin{equation}
\cC_0=\set{c\in \cC| c(a,z)\ge \epsilon(z)a \quad \text{for all } a>0 \text{ and } z\in \ZZ},\label{eq:C2}
\end{equation}
where $\epsilon(z)=x^*(z)^{-1/\gamma}\in (0,1]$. Then $T\cC_0\subset \cC_0$.
\end{prop}

\begin{proof}
Suppose to the contrary that $T\cC_0\not\subset \cC_0$. Then there exists $c\in \cC_0$ such that for some $a>0$ and $z \in \ZZ$, we have $\xi\coloneqq Tc(a,z)<\epsilon(z)a$.

Since $u'$ is strictly decreasing and $\epsilon(z)\in (0,1]$, it follows from \eqref{eq:xi} and $u'(0)=\infty$ that
$$u'(a) \le u'(\epsilon(z)a)<u'(\xi)=\max\set{\E_z \hat{\beta}\hat{R}u'(c(\hat{R}(a-\xi)+\hat{Y},\hat{Z})),u'(a)}.$$
Therefore it must be $u'(a)<\E_z \hat{\beta}\hat{R}u'(c(\hat{R}(a-\xi)+\hat{Y},\hat{Z}))$. Since $u'$ is strictly decreasing and $c\in \cC_0$, we obtain
\begin{align*}
u'(\epsilon(z)a)<u'(\xi)&=\E_z \hat{\beta}\hat{R}u'(c(\hat{R}(a-\xi)+\hat{Y},\hat{Z}))\\
&\le \E_z \hat{\beta}\hat{R}u'(\epsilon(\hat{Z})(\hat{R}(a-\xi)+\hat{Y}))\\
&\le \E_z \hat{\beta}\hat{R}u'(\epsilon(\hat{Z})\hat{R}[1-\epsilon(z)]a).
\end{align*}
Using $u'(c)=c^{-\gamma}$ and $\epsilon(z)=x^*(z)^{-1/\gamma}$, we obtain
\begin{align*}
& x^*(z)<\E_z \hat{\beta}\hat{R}^{1-\gamma}x^*(\hat{Z})[1-x^*(z)^{-1/\gamma}]^{-\gamma}\\
\iff & x^*(z)<\left(1+(\E_z \hat{\beta}\hat{R}^{1-\gamma}x^*(\hat{Z}))^{1/\gamma}\right)^\gamma=\left(1+(K(1-\gamma)x^*)(z)^{1/\gamma}\right)^\gamma,
\end{align*}
which is a contradiction because $x^*$ is a fixed point of $F$ in \eqref{eq:fixedpoint}.
\end{proof}

With all the above preparations, we can prove Theorem~\ref{thm:linear}.

\setcounter{case}{0}
\begin{proof}[Proof of Theoreom \ref{thm:linear}]
Define the sequence $\set{c_n}\subset \cC$ by $c_0(a,z)=a$ and $c_n=Tc_{n-1}$ for all $n\ge 1$. Since $Tc(a,z)\le a$ for any $c\in \cC$, in particular $c_1(a,z)=Tc_0(a,z)\le a=c_0(a,z)$. Since $T:\cC\to \cC$ is monotone, by induction $0\le c_n\le c_{n-1}$ for all $n$ and $c(a,z)=\lim_{n\to\infty}c_n(a,z)$ exists. By Theorem~\ref{thm:exist}, this $c$ is the unique fixed point of $T$ and also the unique solution to the income fluctuation problem \eqref{eq:IF}.

Define the sequence $\set{x_n}\subset \R_{++}^Z$ by $x_0=1$ and $x_n=Fx_{n-1}$, where $F$ is as in \eqref{eq:fixedpoint}. By definition, we have $c_0(a,z)/a=1=x_0(z)^{-1/\gamma}$, so in particular $\limsup_{a\to\infty}c_0(a,z)/a\le x_0(z)^{-1/\gamma}$ for all $z\in \ZZ$. Since $c_n\downarrow c\ge 0$ point-wise, a repeated application of Proposition~\ref{prop:UB} implies that
\begin{equation}
0\le \limsup_{a\to\infty}\frac{c(a,z)}{a}\le \limsup_{a\to\infty}\frac{c_n(a,z)}{a} \le x_n(z)^{-1/\gamma}.\label{eq:barcub}
\end{equation}

\begin{case}[$r(K(1-\gamma))\ge 1$ and $K(1-\gamma)$ is irreducible]
By Proposition~\ref{prop:MPC_seq} we have $x_n(z)\to \infty$ for all $z\in \ZZ$. Letting $n\to\infty$ in \eqref{eq:barcub}, we obtain
$$\lim_{a\to\infty}\frac{c(a,z)}{a}=0.$$
\end{case}
\begin{case}[$r(K(1-\gamma))<1$]
By Proposition~\ref{prop:MPC_seq} we have $x_n(z)\to x^*(z)$, where $x^*$ is the unique fixed point of $F$ in \eqref{eq:fixedpoint}. Letting $n\to\infty$ in \eqref{eq:barcub}, we obtain
\begin{equation}
\limsup_{a\to\infty}\frac{c(a,z)}{a}\le x^*(z)^{-1/\gamma}.\label{eq:limsup}
\end{equation}
On the other hand, a repeated application of Proposition~\ref{prop:LB} implies that $c_n(a,z)\ge x^*(z)^{-1/\gamma}a$ for all $a>0$ and $z\in \ZZ$. Since $c_n\to c$ point-wise, letting $n\to\infty$, dividing both sides by $a>0$, and letting $a\to\infty$, we obtain
\begin{equation}
\liminf_{a\to\infty}\frac{c(a,z)}{a}\ge x^*(z)^{-1/\gamma}.\label{eq:liminf}
\end{equation}
Combining \eqref{eq:limsup} and \eqref{eq:liminf}, we obtain $\lim_{a\to\infty} c(a,z)/a=\bar{c}(z)=x^*(z)^{-1/\gamma}$. \qedhere
\end{case}
\end{proof}

\begin{proof}[Proof of Proposition \ref{prop:noinc}]
Since the proof is similar to \citet*[Proposition 1]{Toda2019JME}, we only provide a sketch.

If $V(a,z)$ denotes the value function, then by homotheticity we can show $V(\lambda a,z)=\lambda^{1-\gamma}V(a,z)$ for any $\lambda>0$. Setting $(a,\lambda)=(1,a)$, we obtain $V(a,z)=V(1,z)a^{1-\gamma}\eqqcolon \frac{x(z)}{1-\gamma}a^{1-\gamma}$ for some $x(z)>0$. The Bellman equation then implies
\begin{equation*}
\frac{x(z)}{1-\gamma}a^{1-\gamma}=\max_{0\le c\le a}\set{\frac{c^{1-\gamma}}{1-\gamma}+\E_z\hat{\beta}\frac{x(\hat{Z})}{1-\gamma}[\hat{R}(a-c)]^{1-\gamma}}.
\end{equation*}
Maximizing the right-hand side over $c$, elementary calculus shows
\begin{equation*}
c=\frac{a}{1+(\E_z\hat{\beta}\hat{R}^{1-\gamma}x(\hat{Z}))^{1/\gamma})}.
\end{equation*}
Substituting this consumption policy into the Bellman equation and comparing coefficients, after some algebra we obtain
\begin{equation*}
x(z)=\left(1+(\E_z\hat{\beta}\hat{R}^{1-\gamma}x(\hat{Z}))^{1/\gamma})\right)^\gamma.
\end{equation*}
Letting $x=(x(1),\dots,x(Z))\in \R_+^Z$, the above equation is exactly \eqref{eq:fixedpoint}, which has a solution if and only if $r(K(1-\gamma))<1$ by Proposition~\ref{prop:MPC_seq}. Under this condition, we can verify the transversality condition as in \cite*{Toda2019JME}. Therefore the zero income model has a solution if and only if $K(1-\gamma)<1$, in which case the value and consumption functions are given by \eqref{eq:noinc}.
\end{proof}

The proof of Theorem~\ref{thm:complete} follows from the same idea as Theorem~\ref{thm:linear} by considering each diagonal block separately.

\setcounter{case}{0}
\begin{proof}[Proof of Theorem~\ref{thm:complete}]
Since $K=K(1-\gamma)$ is a nonnegative matrix (with entries that are potentially infinite), the map $F$ in \eqref{eq:fixedpoint} is monotone and therefore $\set{x_n}_{n=0}^\infty$ monotonically converges to some $x^*\in [1,\infty]^Z$. To characterize $x^*(z)$ and $\bar{c}(z)$, we consider two cases.

\begin{case}[There exist $j$, $\hat{z}\in \ZZ_j$, and $m\in \N$ such that $K^m_{z\hat{z}}>0$ and $r(K_j)\ge 1$]
Define the block diagonal matrix $\tilde{K}=\diag(K_1,\dots,K_J)$ and the sequence $\set{\tilde{x}_n}_{n=0}^\infty\subset [0,\infty]^Z$ by $\tilde{x}_0=1$ and iterating \eqref{eq:fixedpoint}, where $K$ is replaced by $\tilde{K}$. Since $K\ge \tilde{K}\ge 0$, clearly $x_n\ge \tilde{x}_n\ge 1$ for all $n$. Since by definition $\tilde{K}$ is block diagonal with each diagonal block irreducible, by Lemma~\ref{lem:xn} we have $\tilde{x}_n(z)\to\infty$ as $n\to\infty$ if and only if there exists $j$ such that $z\in \ZZ_j$ and $r(K_j)\ge 1$. (Although Lemma~\ref{lem:xn} assumes the entries of $K$ are finite, the infinite case is similar.) Replacing the vector $1$ in \eqref{eq:fixedpoint} by $0$ and iterating, we obtain
$$x_{m+n}\ge K^mx_n\ge K^m\tilde{x}_n.$$
Therefore if there exist $j$, $\hat{z}\in \ZZ_j$ and $m\in \N$ such that $K^m_{z\hat{z}}>0$ and $r(K_j)\ge 1$, then
$$x_{m+n}(z)\ge K^m_{z\hat{z}}\tilde{x}_n(\hat{z})\to\infty$$
as $n\to\infty$, so $x^*(z)=\infty$. In this case we obtain $\bar{c}(z)=0$ by the same argument as in the proof of Proposition~\ref{prop:UB}.
\end{case}
\begin{case}[For all $j$, either $r(K_j)<1$ or $K^m_{z\hat{z}}=0$ for all $\hat{z}\in \ZZ_j$ and $m\in \N$]
For any $\hat{z}$ such that $K^m_{z\hat{z}}=0$ for all $m$, by \eqref{eq:fixedpoint} the value of $x_n(z)$ is unaffected by all previous $x_k(\hat{z})$ for $k<n$. Therefore for the purpose of computing $x_n(z)$, we may drop all rows and columns of $K$ corresponding to such $\hat{z}$. The resulting matrix has block diagonal entries $K_j$ with $r(K_j)<1$ only, so this matrix has spectral radius less than 1. Therefore by Lemma~\ref{lem:xn}, we have $x_n(z)\to x^*(z)<\infty$ as $n\to\infty$. In this case we obtain $\bar{c}(z)=x^*(z)^{-1/\gamma}$ by the same argument as in the proof of Theorem~\ref{thm:linear}.\qedhere
\end{case}
\end{proof}

\begin{proof}[Proof of Proposition \ref{prop:gamless1}]
If $\gamma=1$, then $r(K(1-\gamma))=r(K(0))<1$ by Assumption~\ref{asmp:spectral}\ref{item:rbeta}. Suppose $\gamma\in (0,1)$. For a nonnegative matrix $A$ and $\theta>0$, let $A^{(\theta)}=(A_{z\hat{z}}^\theta)$ be the matrix of $\theta$-th power. Also, let $\odot$ denote the Hadamard (entry-wise) product. Applying H\"older's inequality, we obtain
\begin{equation*}
\E_{z,\hat{z}} \hat{\beta}\hat{R}^{1-\gamma}=\E_{z,\hat{z}} \hat{\beta}^\gamma (\hat{\beta}\hat{R})^{1-\gamma}\le (\E_{z,\hat{z}} \hat{\beta})^\gamma (\E_{z,\hat{z}} \hat{\beta}\hat{R})^{1-\gamma}.
\end{equation*}
Multiplying both sides by $P_{z\hat{z}}\ge 0$ and collecting into a matrix,  we obtain
\begin{equation*}
K(1-\gamma)\le K(0)^{(\gamma)}\odot K(1)^{(1-\gamma)}.
\end{equation*}
Applying Theorem 1 of \cite*{ElsnerJohnsonDiasDaSilva1988}, we obtain
\begin{equation*}
r(K(1-\gamma))\le r(K(0))^\gamma r(K(1))^{1-\gamma}<1
\end{equation*}
by Assumption~\ref{asmp:spectral}\ref{item:rbeta}.

Next, suppose that there exists $z\in \ZZ$ such that $P_{zz}>0$, $\beta(z,z,\zeta)>0$, and $0<R(z,z,\zeta)<1$ with positive probability. Then $P_{zz}\E_{z,z}\hat{\beta}\hat{R}^{1-\gamma}>1$ for large enough $\gamma>1$. Letting $\tilde{K}$ be the matrix whose $(z,z)$ entry is $P_{zz}\E_{z,z}\hat{\beta}\hat{R}^{1-\gamma}>1$ and all other entries are zero, we obtain $K(1-\gamma)\ge \tilde{K}$ entry-wise. Therefore $r(K(1-\gamma))\ge r(\tilde{K})>1$ by Theorem 8.1.18 of \cite*{HornJohnson2013}.
\end{proof}

\begin{proof}[Proof of Proposition \ref{prop:Bewley}]
\cite*{StachurskiToda2019JET} show that it must be $\beta R<1$ in the stationary equilibrium.

If $R\ge 1$, then $\beta R^{1-\gamma}=(\beta R)R^{-\gamma}<1$. By Example~\ref{exmp:z1}, the asymptotic MPC is $\bar{c}=1-(\beta R^{1-\gamma})^{1/\gamma}\in (0,1)$. Therefore the asymptotic saving rate \eqref{eq:saverateinf} simplifies to
\begin{equation*}
\bar{s}=1-\frac{\bar{c}}{(R-1)(1-\bar{c})}=\frac{(\beta R)^{1/\gamma}-1}{(R-1)(\beta R^{1-\gamma})^{1/\gamma}}\in [-\infty,0)
\end{equation*}
because $\beta R<1$ and $R\ge 1$.

If $R<1$, then the saving rate \eqref{eq:saverate2} becomes
\begin{equation*}
s_{t+1}=1-\frac{(1-R)(1-c/a)+c/a}{\hat{Y}/a}.
\end{equation*}
As $a\to \infty$, we have $c/a\to \bar{c}\in [0,1]$ and $\hat{Y}/a\to 0$. Since $R<1$, it follows that $s_{t+1}\to -\infty$.
\end{proof}

\begin{proof}[Proof of Proposition \ref{prop:Bewley2}]
Since by assumption $\E \beta R^{1-\gamma}<1$, by Example~\ref{exmp:z1} the asymptotic MPC is $\bar{c}=1-(\E \beta R^{1-\gamma})^{1/\gamma}\in (0,1)$. If $\E R\ge 1$, the asymptotic saving rate \eqref{eq:saverateinf} evaluated at $\hat{R}=\E R$ becomes
\begin{equation*}
\bar{s}=1-\frac{\bar{c}}{(\E R-1)(1-\bar{c})}=\frac{\E R(1-\bar{c})-1}{(\E R-1)(1-\bar{c})}.
\end{equation*}
Since $\E R(1-\bar{c})$ is the expected growth rate of wealth for infinitely wealthy agents, if the wealth distribution is unbounded and $\E R(1-\bar{c})>1$, then wealth will grow at the top, which violates stationarity. Therefore in a stationary equilibrium, it must be $\E R(1-\bar{c})\le 1$ and hence $\bar{s} \le 0$.

If $\E R<1$, the proof is identical to the risk-free case (Proposition~\ref{prop:Bewley}).
\end{proof}

%\end{document}

% define local bibliography database

\end{document}